\documentclass[aps,twocolumn,superscriptaddress,showpacs,amsmath,amssymb,nofootinbib]{revtex4}

\usepackage{graphicx}% Include figure files
\usepackage{dcolumn}% Align table columns on decimal point
\usepackage{bm}% bold math
\usepackage{amsthm}
\usepackage{graphicx}
\usepackage{url}
\usepackage{color}
\usepackage{hyperref}
\definecolor{refcolor}{RGB}{0,0,190}
\hypersetup{
    colorlinks,
    citecolor=refcolor,
    filecolor=refcolor,
    linkcolor=refcolor,
    urlcolor=refcolor
}

\begin{document}
\newtheorem{theorem}{Theorem}[section]
\newtheorem{proposition}[theorem]{Proposition}
\newtheorem{lemma}{Lemma}[section]
\newtheorem{corollary}{Corollary}[section]
\newtheorem{definition}{Definition}[section]
\newtheorem{principle}{Principle}[section]
\newtheorem{hypothesis}{Hypothesis}[section]
\newtheorem{property}{Property}[section]
\newtheorem{problem}{Problem}[section]
\newtheorem{objection}{Objection}
\theoremstyle{remark}
\newtheorem{remark}{Remark}[section]
\newtheorem{example}{Example}[section]
\newtheorem{reply}{Reply}

\newenvironment{solution}[2] {\paragraph{Solution to {#1} {#2} :}}{\hfill$\square$}
%\preprint{xxxxxx}

%--------------------------------------------------------
% Defines
   
\def\({\left(}
\def\){\right)}

\newcommand{\tn}{\textnormal}
\newcommand{\ds}{\displaystyle}
\newcommand{\dsfrac}[2]{\displaystyle{\frac{#1}{#2}}}

\newcommand{\statespace}{\mathcal{S}}
\newcommand{\hilbert}{\mathcal{H}}
\newcommand{\vectorspace}{\mathcal{V}}
\newcommand{\mc}[1]{\mathcal{#1}}
\newcommand{\ms}[1]{\mathscr{#1}}

\newcommand{\qmU}{$\mathscr{U}$}
\newcommand{\qmR}{$\mathscr{R}$}
\newcommand{\qmUR}{$\mathscr{UR}$}
\newcommand{\qmDR}{$\mathscr{DR}$}

\newcommand{\R}{\mathbb{R}}
\newcommand{\C}{\mathbb{C}}
\newcommand{\Z}{\mathbb{Z}}
\newcommand{\K}{\mathbb{K}}
\newcommand{\N}{\mathbb{N}}
\newcommand{\Prj}{\mathcal{P}}
\newcommand{\abs}[1]{\left|#1\right|}

\newcommand{\de}{\operatorname{d}}
\newcommand{\tr}{\operatorname{tr}}
\newcommand{\im}{\operatorname{Im}}

\newcommand{\ie}{\textit{i.e.}}
\newcommand{\vs}{\textit{vs.}}
\newcommand{\eg}{\textit{e.g.}}
\newcommand{\etc}{\textit{etc}}
\newcommand{\etal}{\textit{et al.}}

\newcommand{\Span}{\tn{span}}
\newcommand{\pde}{PDE}
\newcommand{\U}{\tn{U}}
\newcommand{\SU}{\tn{SU}}
\newcommand{\GL}{\tn{GL}}

\newcommand{\schrod}{Schr\"odinger}
\newcommand{\vonneum}{Liouville - von Neumann}
\newcommand{\ks}{Kochen-Specker}
\newcommand{\leggarg}{Leggett-Garg}
\newcommand{\bra}[1]{\langle#1|}
\newcommand{\ket}[1]{|#1\rangle}
\newcommand{\braket}[2]{\langle#1|#2\rangle}
\newcommand{\expectation}[1]{\langle#1\rangle}
\newcommand{\Herm}{\tn{Herm}}
\newcommand{\Sym}[1]{\tn{Sym}_{#1}}

\newcommand{\btimes}{\boxtimes}
\newcommand{\btimess}{{\boxtimes_s}}

\newcommand{\x}{\mathbf{x}}
\newcommand{\M}{\mathbb{E}_3}
\newcommand{\D}{\mathbf{D}}
\newcommand{\dn}{\tn{d}}
\newcommand{\db}{\mathbf{d}}
\newcommand{\n}{\mathbf{N}}
\newcommand{\V}[1]{\mathbb{V}_{#1}}
\newcommand{\F}[1]{\mathcal{F}_{#1}}
\newcommand{\Fvacuumfield}{\widetilde{\mathcal{F}}^0}
\newcommand{\nD}[1]{|{#1}|}
\newcommand{\Lin}{\mathcal{L}}
\newcommand{\End}{\tn{End}}
\newcommand{\vbundle}[4]{{#1}\to {#2} \stackrel{\pi_{#3}}{\to} {#4}}
\newcommand{\vbundlex}[1]{\vbundle{V_{#1}}{E_{#1}}{#1}{M_{#1}}}
\newcommand{\rep}{\rho_{\scriptscriptstyle\btimes}}

\newcommand{\Hint}{H_{\tn{int}}}

\newcommand{\quot}[1]{``#1''}

\def\sref #1{\S\ref{#1}}

\newcommand{\dBB}{de Broglie--Bohm}
\newcommand{\dBBt}{{\dBB} theory}
\newcommand{\pwt}{pilot-wave theory}
\newcommand{\PWT}{PWT}
\newcommand{\NRQM}{{\textbf{NRQM}}}

\newcommand{\image}[3]{
\begin{center}
\begin{figure*}[!ht]
\includegraphics[width=#2\textwidth]{#1}
\caption{\small{\label{#1}#3}}
\end{figure*}
\end{center}
}

\title{A representation of the wave function on the three-dimensional space}

\author{Ovidiu Cristinel Stoica}
% \homepage{http://www.Second.institution.edu/~Charlie.Author}
\affiliation{
 Dept. of Theoretical Physics, NIPNE---HH, Bucharest, Romania. \\
	Email: \href{mailto:cristi.stoica@theory.nipne.ro}{cristi.stoica@theory.nipne.ro},  \href{mailto:holotronix@gmail.com}{holotronix@gmail.com}
	}%

\date{\today}% It is always \today, today,
             %  but any date may be explicitly specified

\begin{abstract}
One of the major concerns of Schr\"odinger, Lorentz, Einstein, and many others about the wave function is that it is defined on the $3\mathbf{N}$-dimensional configuration space, rather than on the $3$-dimensional physical space. This gives the impression that quantum mechanics cannot have a three-dimensional space or spacetime ontology, even in the absence of quantum measurements. In particular, this seems to affect interpretations which take the wave function as a physical entity, in particular the many worlds and the spontaneous collapse interpretations, and some versions of the pilot wave theory.

Here, a representation of the many-particle states is given, as multi-layered fields defined on the $3$-dimensional physical space. This representation is equivalent to the usual representation on the configuration space, but it makes it explicit that it is possible to interpret the wave functions as defined on the physical space. As long as only unitary evolution is involved, the interactions are local. I intended this representation to capture and formalize the non-explicit and informal intuition of many working quantum physicists, who, by considering the wave function sometimes to be defined on the configuration space, and sometimes on the physical space, may seem to researchers in the foundations of quantum theory as adopting an inconsistent view about its ontology. This representation does not aim to solve the measurement problem, and it allows for Schr\"odinger cats just like the usual one. But it may help various interpretations to solve these problems, through inclusion of the wave function as (part of) their primitive ontology.

In an appendix, it is shown how the multi-layered field representation can be extended to quantum field theory.
\end{abstract}

%\keywords{Suggested keywords}%Use showkeys class option if keyword
                              %display desired
\maketitle

%------------------------------------------------------------%
\section{Introduction}
\label{s:intro}

%------------------------------------------------------------%
\subsection{The problem, and various attempts to solve it}
\label{s:intro_problem}

Perhaps not as important as the measurement problem and the problem of the emergence of the quasi-classical world at macroscopic level, the fact that the wave function of many particles is defined on the $3\mathbf{N}$-dimensional configuration space, rather than on the $3$-dimensional physical space, has been a point of concern for many physicists, ever since the discovery of quantum mechanics.

In a letter to {\schrod}, Lorentz wrote (\cite{Przibram1967LettersWaveMechanics}, p. 44):

\begin{quote}
If I had to choose now between your wave mechanics and the matrix mechanics, I would give the preference to the former, because of its greater intuitive clarity, so long as one only has to deal with the three coordinates $x$, $y$, $z$. If, however, there are more degrees of freedom, then I cannot interpret the waves and vibrations physically, and I must therefore decide in favor of matrix mechanics.
\end{quote}

Similarly, Einstein wrote to Ehrenfest that ``{\schrod} is, in the beginning, very captivating. But the waves in $n$-dimensional coordinate space are indigestible'' (\cite{Howard1990EinsteinWorriesQM}, August 28, 1926). He expresses similar concerns in letters to Lorentz (May 1, 1926), Ehrenfest (June 18, 1926), Lorentz (June 22, 1926), Sommerfeld (August 21, 1926), Lorentz (February 16, 1927) \cite{Howard1990EinsteinWorriesQM}.

{\schrod} in particular seemed disappointed that he couldn't interpret the wave function as a physical entity in space and time, in the same way de Broglie proposed for single particles (\cite{BacciagaluppiValentini2009SolvayConference}, p. 447):

\begin{quote}
Of course this use of the $q$-space is to be seen only as a mathematical tool, as it is often applied also in the old mechanics; ultimately, in this version also, the process to be described is one in space and time. In truth, however, a complete unification of the two conceptions has not yet been achieved. Anything over and above the motion of a single electron could be treated so far only in the \emph{multi}-dimensional version...
\end{quote}

He expressed more similar worries in \cite{Przibram1967LettersWaveMechanics,Przibram2011LettersWaveMechanics,Schrodinger1926QuantisierungAlsEigenwertproblem,Schrodinger2003CollectedPapersWaveMechanics}. Einstein had similar doubts \cite{Howard1990EinsteinWorriesQM,FineBrown1988ShakyGameEinsteinRealismQT}, and Bohm too \cite{Bohm2004CausalityChanceModernPhysics}.

These concerns continue to exist even today, as can be seen {\eg} in \cite{Monton2006QM3Nspace,NeyAlbert2013WavefunctionMetaphysics,Norsen2017FoundationsQM,Gao2017MeaningWavefunction,Maudlin2019PhilosophyofPhysicsQuantumTheory}, where more discussions of this problem can be found. In the rest of this subsection, I will briefly review some modern attempts to either avoid the problem, or to reduce its impact. The rest of the article will present the detailed mathematical construction of a representation of the wave function on the three-dimensional ($3$D) space, whose Hamiltonian evolution is local, and its implications will be discussed.

The most direct way, although probably not the easiest one to accept, is to take seriously the idea that the configuration space is indeed the truly physical arena on which the wave function is defined. The fact that it looks three-dimensional should be enough to its inhabitants ``if they don't look too closely'' \cite{DavidAlbert1996ElementaryQuantumMetaphysics,Vaidman2016AllIsPsi}. This is probably the intuition behind many working quantum theorists, who may seem to some of those concerned with the foundations as not realizing that there is a problem with the ontology of the wave function (see {\eg} \cite{Norsen2017FoundationsQM}, p. 134-135). I am not sure that they ignore the problem, they may simply find the wave function on the configuration space acceptable, it would not be the first time physics challenges our classical intuitions. Or maybe they have an intuitive view like the one that the representation presented in this article will make rigorous. 

But nevertheless, the position that a wave function defined on the configuration space is acceptable has been well developed at the philosophical-foundational level too \cite{DavidAlbert1996ElementaryQuantumMetaphysics,Loewer1996HumeanSupervenience,DAlbertBLoewer1996TailsSchrodingerCat,Ney2012Status3DQuantum,Ney2013OntologicalReductionWavefunctionOntology,North2013StructureOfQuantumWorld,Albert2019How2TeachQM}, in particular in connection to the many worlds interpretation \cite{Barrett1999QMofMindsWorlds,SEP-Vaidman2002MWI,Wallace2002WorldsInMWI,Wallace2003EverettAndStructure,BrownWallace2005BohmVsEverett,Barrett2017TypicalWorlds}. This position was criticized in \cite{Monton2002WavefunctionOntology,Monton2006QM3Nspace,Maudlin2007CompletenessSupervenienceOntology,Allori2008CommonBMandGRW,Maudlin2010CanTheWorldBeOnlyWavefunction,Maudlin2013TheNatureOfTheQuantumState,Monton2013Against3NSpace,Chen2017OurFundamentalPhysicalSpace,Emery2017AgainstRadicalQuantumOntologies,Maudlin2019PhilosophyofPhysicsQuantumTheory}. In particular, in \cite{Lewis2004LifeConfigurationSpace} the main argument of \cite{DavidAlbert1996ElementaryQuantumMetaphysics} is questioned, but the position that the wave function should also be a $3$D-object is defended on the grounds that the Hamiltonian is invariant only to the isometries of space, and not of those of the configuration space. More debates about the wave function ontology can be found in \cite{NeyAlbert2013WavefunctionMetaphysics}.

Another proposed possibility is to interpret the wave function as a \emph{multi-field}, {\ie} something like a field, but which assigns properties to regions of space \cite{Forrest1988QuantumMetaphysics,Belot2012QuantumStatesPrimitiveOntologists,Chen2017OurFundamentalPhysicalSpace,Chen2018IntrinsicStructureQM,Chen2019RealismAboutWavefunction,HubertRomano2018WavefunctionMultifield}. One of its advantages is to give a faithful representation of indistinguishable particles, one that doesn't work out naturally in the configuration space.

Starting from Bohmian Mechanics, Norsen developed a theory of local beables, in which each particle has an individual wave function on space, and entanglement is maintained through some additional fields \cite{Norsen2010TheoryOfLocalBeables}. Its purpose was demonstrative for the possibility to avoid a wave function on the configuration space, and works for spinless nonrelativistic particles whose wave functions are everywhere analytic.

Another somewhat conciliating position between the configuration space and the physical space is to treat the wave function as defined on the physical space for separate particles, and to use their reduced density matrices for the entangled particles \cite{WallaceTimpson2010QMOnSpacetime}. An extension to the relativistic case and quantum field theory is done in \cite{Swanson2018RelativisticSpacetimeStateRealist}.

%------------------------------------------------------------%
\subsection{Summary of the proposed solution}
\label{s:intro_solution}

By contrast to proposed ontologies on the configuration space or on subsets of the physical space, the one following from the representation presented here is defined pointwisely, in a local separable manner, on the $3$D space. And by contrast to ontologies that do not contain the necessary information to completely recover the wave function, the representation proposed here is faithful. Note that it does not necessarily follow that this representation reflects the true ontology, so I will take a neutral position and avoid making such a claim in either direction. By this, I hope to allow for the possibility to adopt it freely in various interpretations, according to their own ontological commitments, regardless whether the wave function is seen as ontic, as epistemic, as nomological, or various combinations.

Section \sref{s:def} introduces some basic definitions and notations, and recalls some known facts about fiber bundles and nonrelativistic quantum mechanics, which will be used in the subsequent sections. Section \sref{s:st_rep} builds gradually the $3$D space representation of the wave functions, which is shown to be isomorphic to the usual representation as functions on the configuration space. The resulting \emph{multi-layered field representation} is shown to be defined on the $3$D space.

Now I will sketch, in a rather handwaving manner, the proposed solution, which is effectively constructed mathematically in section \sref{s:st_rep}, and analyzed in the subsequent sections. The first remark is that all quantum states are superpositions of tensor products of states defined on the $3$D space. In terms of wave functions on the configuration space, this already suggests the following:
\begin{enumerate}
	\item 
Separable states of the form $\Psi(\x_1,\ldots,\x_\n)=\psi_1(\x_1)\ldots\psi_\n(\x_\n)$ consist of $\n$ wave functions defined on the $3$D space.
	\item 
Since general states are linear combinations of separable states, solving the problem for separable states makes possible to take the representation of unseparable states as linear combinations of representations of separable states.
\end{enumerate}

This is probably already the intuition of some researchers, but it needs to be constructed rigorously, and this is not as straightforward as we may hope.

A first major problem appears already in the first step: as combined in the product state $\Psi$, the wave functions $\psi_j$ depend on $\n$ distinct positions $\x_1,\ldots,\x_\n$. One can't replace all their arguments $\x_j$ by the same position $\x$, since this will lead to the loss of the most information contained in $\Psi$. We need a way to keep all the wave functions $\psi_j$ distinct in a collection $\(\psi_1,\ldots,\psi_\n\)$. It seems natural to consider them as parts of a highly dimensional field on the $3$D space, but there is a problem: the decomposition $\Psi(\x_1,\ldots,\x_\n)=\psi_1(\x_1)\ldots\psi_\n(\x_\n)$ is not uniquely determined. We can multiply by different constants each of the wave functions $\psi_j$, and get the same result. This complication forces us to introduce an equivalence relation, or a gauge symmetry, and to take as representation a field on the $3$D space which is the result of factoring this equivalence out of the collections of the form $\(\psi_1,\ldots,\psi_\n\)$. This will be done in \sref{s:separable}. Then, another major problem arises, since factoring out the equivalence seems, at first sight, to make the construction make sense only globally on the $3$D space. So the representation seems to not be \emph{local separable}. This problem will be addressed, in the same way in which, in gauge theory, global symmetries are promoted to local symmetries, in \sref{s:local_separability}. The result is a representation in terms of fields with a large number of components, on the $3$D space. Such a field belongs to a space of functions that will be called \emph{layer} in the following.

Once the first step of the representation is realized, we need to extend the representation to entangled or nonseparable states. One may be tempted to take them as linear superpositions of representations of separable states. But this doesn't work immediately, since such linear combinations don't commute with the equivalence relation used to represent separable states, and even if they would commute, the result would be again a separable state in the same layer, which is not the nonseparable state we want to represent. For this reason, we need to represent nonseparable states as superpositions of representations of separable states from different layers, hence the name \emph{multi-layered field representation}. To obtain it, we need to avoid some pitfalls, which will be explained along the way. The way to do it is to first represent the separable states forming a basis of the Hilbert space, and then construct the vector bundle freely generated by these. This construction will be done gradually in \sref{s:more_separable} and \sref{s:nonseparable}. I summarize the result in Figure \ref{multi-layered.png}, where an example of a three-particle state is given.

\image{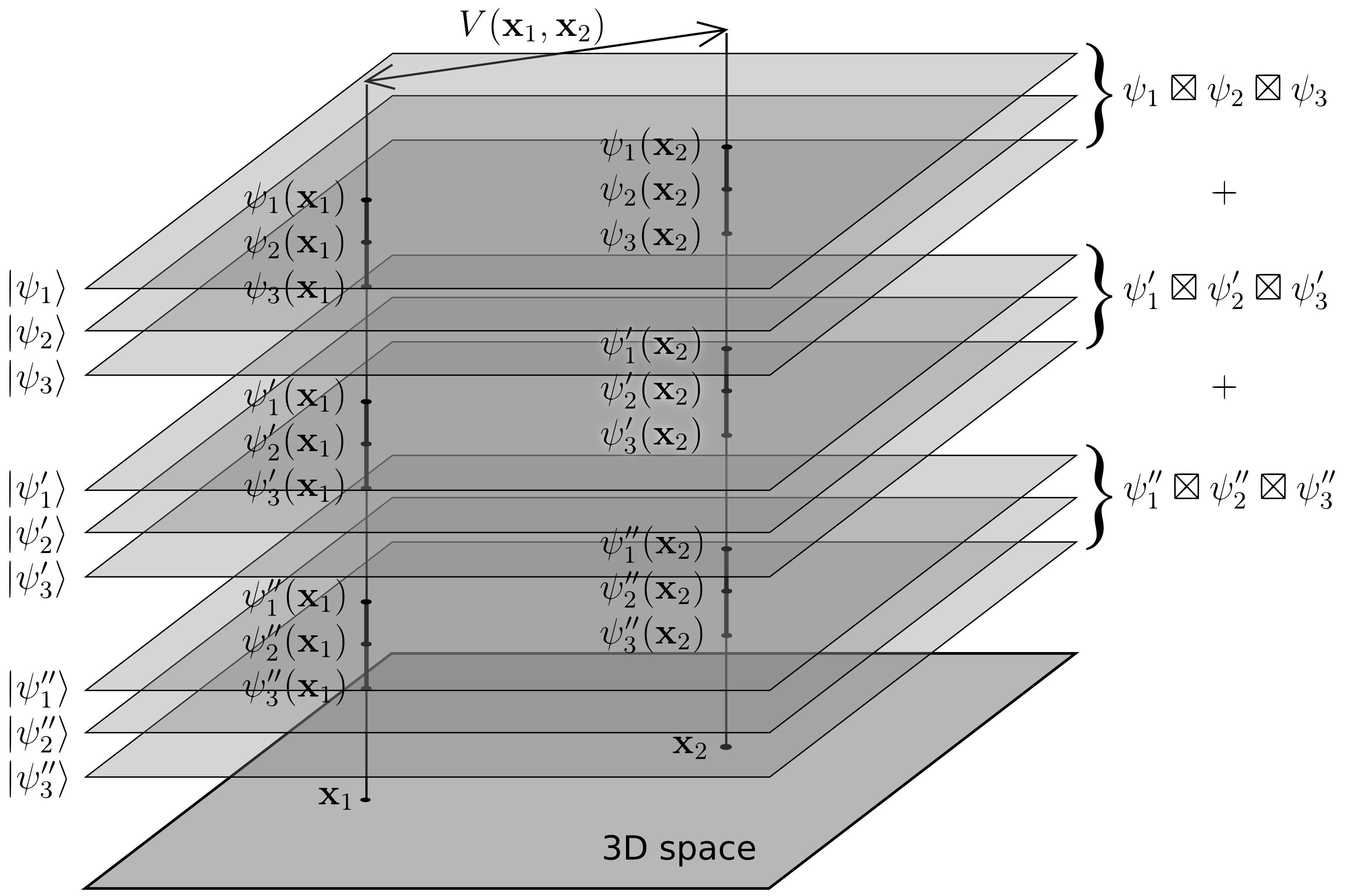}{0.7}{Depiction of the multi-layered field representation. Separable states are represented as fields defined on the $3$D space, in the same layer. Each layer corresponds to product states, and has a certain gauge symmetry (\sref{s:separable}). Sums of separable states are represented as sums of fields representing separable states (\sref{s:more_separable}). The interactions happen between pointwise values of fields, and depend on the $3$-dimensional distance, not on the configuration space $3\n$-dimensional distance. If the interactions propagate locally, then the unitary evolution is local (section \sref{s:locality}).}

From physical point of view, each layer contains $\n$ $3$D wave functions, corresponding to $\n$ particles. There is nothing more here about the particles, than their $3$D wave functions. They propagate and interact according to the {\schrod} equation. If the layers are chosen corresponding to eigenspaces of the Hamiltonian, then the $3$D wave functions remain in the same layer. But the principle of superposition allows the existence of other independent such $3$D wave functions, which evolve independently. They will be in a different layer, and the $3$D wave functions from one layer will ignore those in other layers. Note that the layers can be defined differently, if we choose a different basis to start with, and if this basis is not an eigenbasis of the Hamiltonian, the $3$D wave functions will move from one layer to another.

In section \sref{s:locality}, we take a look at the unitary dynamics of the multi-layered fields, and we find that it is local, as long as unitary evolution is not interrupted by a state vector reduction (collapse of the wave function). The interpretation of this representation is given in section \sref{s:multilayer}. If measurements are taking place, especially on entangled states, nonlocal correlations appear, so section \sref{s:EPR} will analyze this, exemplifying with the EPR experiment.

In section \sref{s:implications} I propose the multi-layered field representation as a way to provide a $3$D space (primitive) ontology for various interpretations of quantum mechanics. 
The construction proposed here doesn't aim to solve the measurement problem, and the problem of the emergence of the quasi-classical world at macroscopic level. 
It's only purpose is to replace the wave function on the configuration space by something completely equivalent, but defined on the $3$D space.
To solve these problems, an extension of quantum mechanics, usually called \emph{interpretation}, is needed. In particular, in a {\pwt}, the wave functions corresponding to particles will be supplemented by point-particles, which allow us to solve these foundational problems. Possibly, the Many Worlds Interpretation may not require more than the wave function evolving unitarily, in which case the branches corresponding to the many worlds will simply be layers. But I will not discuss this in this article, limiting myself to simply provide the $3$D space representation.

Appendix \sref{s:vbundle_math} contains the more general and rigorous definition of vector bundles. Appendix \sref{s:quantum_fields} describes the extension of multi-layered field representation to quantum field theory, exemplifying with scalar fields.

In section \sref{s:objections} I try to anticipate and discuss potential objections to the multi-layered field representation as a representation on the $3$D space. Such a dedicated section may be unusual, but I considered that, given the resilience of the problem and the widespread opinion that it is unsolvable, it is normal for objections to exist, and it is better to address them directly, rather than leaving room for misunderstandings.

%------------------------------------------------------------%
\section{Basic definitions and notations}
\label{s:def}

This section recalls or introduces some basic definitions and notations concerning fiber bundles, the configuration space, and the wave functions, which will be used in the next sections.

%------------------------------------------------------------%
\subsection{Operations with vector bundles}
\label{s:vbundle}

An important notion is that of a \emph{field}, in particular of \emph{vector field}. The vector fields will not necessarily have as values 3D vectors, because their values can be, depending on the situation, spinors, vectors, tensors, vectors from internal spaces used in gauge theory \etc. All these are still classical fields. What I mean by ``fields'' here will be used to represent states, being in fact vector-valued wave functions, so they should not be confused with \emph{quantum field operators}, which are defined on the vector spaces consisting of such fields. 
Their purpose is to be used to represent on the 3D Euclidean space $\M\cong\R^3$ quantum states or wave functions from {\NRQM}, which are usually defined on the configuration space.

Vector fields on a base space $M$ are taken here to simply be functions defined on $M$, and valued in some fixed vector space $V$. As an example, one can think at the electric field, which associates to each point $\x$ of $\M$ a $3$D vector $\mathbf{E}(\x)\in\R^3$.

Vector fields of the same type can be added, and can be multiplied by scalars. Because of this, they form vector spaces. But other operations are possible, resulting in different types of vector fields. These new types of vector fields are best understood mathematically if we take care to specify in which vector bundle each vector field is. 
Then, one can define operations that take vector bundles and give a different vector bundle. The results of operations with vector fields will be then vector fields from the vector bundle constructed like this.

Some of the operations with vector bundles preserve the space on which the vector fields are defined (the base manifold), while others don't. I will be interested in those operations which, when applied to bundles over the 3D Euclidean space $\M$, result again in bundles over $\M$. The representation I will construct in this article will use only such operations, this ensuring that the usual wave functions over the configuration space can be represented as wave functions or fields over $\M$.

In the following, all vector spaces will be considered to be complex.
I will give now a simpler definition of a vector bundle (the ``trivial vector bundle''), which will be enough for the purpose of this article. The more general definition is \ref{def:vector_bundle_math} from Appendix section \sref{s:vbundle_math}.

\begin{definition}
\label{def:vector_bundle_trivial}
Let $M$ a topological space (which here will in general be the $3$D space or the configuration space).
Let $V$ be a complex $k$-dimensional vector space.
We can use $M$ and $V$ to define a complex \emph{trivial vector bundle} $V \to E \stackrel{\pi}{\to} M$, which consists of the following:
\begin{enumerate}
	\item 
	The vector space $V$, called the \emph{typical fiber}.
	\item 
	The \emph{base space} $M$, and the \emph{total space} $E=M\times V$.
	\item 
	The \emph{bundle projection} $\pi:E\to M$, $\pi(x,v)=\{x\}$.
	\item 
	For each $x\in M$, we define the \emph{fiber} over $x$, $V_x:=\pi^{-1}(x)=\{x\}\times V$, and endow it with the vector space structure of $V$.
\end{enumerate}
\end{definition}

A \emph{vector field} $v:M\to V$ can be seen as a section of the total space of the bundle, $E=M\times V$, such that to any $x\in M$, there corresponds a unique point $(x,v(x))\in E$. So, in terms of vector bundles, a vector field defined as the mapping $x\mapsto(x,v(x))$ associates to each $x\in M$ a vector in the fiber $V_x$. This contrasts to the usual interpretation that $v(x)\in V$.
We denote the collection of all vector fields over $M$ by $\Gamma(E)$. Since we can multiply them by complex numbers, and we can add them, and the result is always another vector field, $\Gamma(E)$ is also a vector space. When $V$ has a \emph{Hermitian scalar} product $h$, and $M$ has a volume form $\de x$ which allows us to integrate, we can define a scalar product between two vector fields by
\begin{equation}
\label{eq:vec_field_scalar_prod}
\langle u,v\rangle :=\int_M h(u(x),v(x))\de x. 
\end{equation}
Normally, we work with a vector subspace of $\Gamma(E)$ which satisfies some nice conditions, for example differentiability, or that the scalar product is always finite. In particular, we can choose our vector space of vector fields to form a Hilbert space, for example by keeping only the square-integrable vector fields.

When one or more vector spaces are given, we can use various operations to build other vector spaces. Such operations include
\begin{enumerate}
\item
The \emph{dual} $V^\ast$ of a vector space $V$, consisting of the $\C$-linear maps $f:V\to\C$.
	\item 
The \emph{direct sum}, which associates to a collection of vector spaces $V_j$, $j\in J$, the vector space $\bigoplus_{j\in J} V_j$.
	\item 
The \emph{tensor product}, which associates to a collection of complex vector spaces $V_j$, $j\in J$, the vector space $\bigotimes_{j\in J} V_j$, where the tensor product is taken over $\C$.
	\item 
The \emph{quotient vector space} $V/W$, defined by a vector subspace $W$ of $V$, by the equivalence relation $v_1\sim v_2$ iff $v_2-v_1\in W$. Then, $V/W$ has as vectors the equivalence classes $[v]$ determined by $\sim$.
	\item 
The \emph{tensor algebra}, $\mathcal{T}(V):=\bigoplus_{j\in\N} \(\bigotimes^j V\)$.
	\item 
The \emph{symmetric algebra}, $\Sym+(V)$, obtained by symmetrizing all tensor products in $\mathcal{T}(V)$.
	\item 
The \emph{anti-symmetric}, or \emph{exterior algebra} $\Sym-(V)$, obtained by anti-symmetrizing all tensor products in $\mathcal{T}(V)$.
	\item 
The vector space of \emph{$\C$-linear endomorphisms} between two complex vector spaces $V,W$, $\End(V,W)$. This vector space is canonically isomorphic to $V^*\otimes W$.
	\item
Various combinations of these operations.
\end{enumerate}

There are two main ways to extend these operations with vector spaces to operations with vector bundles. One way is to apply them to vector spaces of vector fields from the bundles. In particular, the tensor product of vector fields from bundles is normally used when constructing many-particle Hilbert spaces out of the Hilbert spaces of single particles. The resulting vector fields can not be usually interpreted as vector fields on $M$, but on larger base spaces. The other way is to apply the operations with vector spaces at the fiber level, which is normally done in the geometry of fiber bundles. In this case, the operations preserve the base manifold. These two ways lead in general to different results.

To see the difference, consider two vector bundles over the same base manifold $M$, $\vbundle{V_1}{E_1}{1}{M}$ and $\vbundle{V_2}{E_2}{2}{M}$.
The direct sum of the vector spaces of their fields, $\Gamma(\vbundle{V_1}{E_1}{1}{M})\oplus\Gamma(\vbundle{V_2}{E_2}{2}{M})$, can be identified with the vector space of fields from a bundle over $M$, whose fibers over points $x\in M$ are the direct sums of the fibers, $V_{1x}\oplus V_{2x}$.
This works because, given two fields $s_1$ and $s_2$ from the two bundles, they are defined on the same space $M$, and can be recovered by knowing $(s_1(x),s_2(x))$ at all $x$.

On the other hand, in the case of the tensor product, the things are dramatically different.
The tensor product of two fields from the same two vector bundles, seen as elements of the vector spaces of vector fields from these bundles, is defined on $M\times M$, and not on $M$, although it is valued in $V_1\otimes V_2$. The main difference in quantum mechanics is that we have to use such tensor products of spaces of fields from bundles.
The main step in the construction presented here is to represent this tensor product of vector fields in terms of vector bundle operations resulting in bundles over the same base manifold, in our case $\M$. The idea is simple, but its realization is not.

%------------------------------------------------------------%
\subsection{Classical configuration space for $\n$ particles}
\label{s:class_conf}

Consider first the case of a classical point-particle, having a definite position $\x(t)\in \M$, where $\M=\R^3$ stands for the physical 3D space, and $\x=(x,y,z)\in\M$. 
If the particle's state is characterized only by is its position, its configuration space is $\M$. But it may have different states even if it has the same position, and I denote by $\D$ the set of these possible states. $\D$ may include internal states and the spin state along a fixed axis.
In this case, the configuration space of the particle is $\M\times\D$. If the particle's state is characterized only by position, $\D$ will have only one element, so that $\M\times\D\cong\M$. Each set $\D$ may consist of different kinds of states, in which case $\D$ will be the Cartesian product of other sets, but in most cases it will be simply denoted as a single set for each particle.

We will consider that the type of a particle is completely determined by $\D$, so we will simply say that the \emph{type of the particle} is $\D$.

In the case of $\n$ particles of types $\D_1\ldots,\D_\n$, the configuration space is 
\begin{equation}
\label{eq:config_N}
\underbrace{\M\times\ldots\times \M}_{\tn{$\n$ times}} \times \D_1 \times \ldots \times \D_\n = \M^\n \times \prod_{j=1}^\n\D_j.
\end{equation}

Let $\nD\D$ denote the cardinal of the set $\D$. Then, $\nD{\prod_{j=1}^\n\D_j}=\prod_{j=1}^\n\nD{\D_j}$.

%------------------------------------------------------------%
\subsection{Wave functions of $\n$ particles}
\label{s:wave_n}

In \emph{nonrelativistic  quantum mechanics} (\NRQM), the result of the quantization is that instead of having $\n$ classical point-particles, we end out with wave functions defined on the configuration space of $\n$ particles. The values of these wave functions are usually complex numbers. Even if there are particles whose wave functions are real, in order to give a homogeneous treatment I will consider them complex, which is justified by $\R\subset\C$.

After quantization, classical systems of $\n$ point-particles of types $\D_1\ldots,\D_\n$, where $\D_j = \{d_{j}^{1},\ldots,d_{j}^{\nD{\D_j}}\}$ for each $j\in\{1,\ldots,\n\}$, are replaced by quantum system. A quantum state is represented by the \emph{wave function}	

\begin{equation}
\label{eq:universal_wavefunction}
\begin{split}
&\psi:\M^\n \times \prod_{j=1}^\n\D_j \to \C \\
&\psi(\x_1,\ldots,\x_\n,\dn_{1},\ldots,\dn_{\n}) \in \C, \\
\end{split}
\end{equation}
where $\dn_j\in\D_j$ for all $j\in\{1,\ldots,\n\}$.

The set of all wave functions of the form \eqref{eq:universal_wavefunction} forms a vector space, denoted here by $\V{\M^\n,\D_1,\ldots,\D_\n}$. Note that we can include here also the distributions (like Dirac's delta function $\delta(\x)$). We can choose to work with a subspace, for example of the square-integrable wave functions, but we will define this later.

In the particular case when the $\n$ particles are characterized only by positions, and have no additional degrees of freedom, we denote the vector space of wave functions by $\V{\M^\n}$. In the following, this space of scalar functions will also be used to represent the spatial degrees of freedom of wave functions having more degrees of freedom.

%------------------------------------------------------------%
\subsection{Vector-valued wave functions on $\M^\n$}
\label{s:reduction_to_m_n}

According to equation \eqref{eq:universal_wavefunction}, the configuration space of a single particle of type $\D$ is $\M\times\D$, hence for a wave function $\psi$ of such a particle, $\psi:\M\times\D\to\C$. So even a single particle is defined on the configuration space $\M\times\D = \bigsqcup_{j\in\D}\M$, rather than on the 3D space $\M$. But this is just a matter of representation, since we can regard $\psi$ as having $\nD \D$ components, each defined on $\M$,
\begin{equation}
\label{eq:wave_components}
\begin{cases}
\psi^1(\x) := \psi(\x, d^1) \\
\ldots \\
\psi^{\nD \D}(\x) := \psi(\x, d^{\nD \D}). \\
\end{cases}
\end{equation}

The components $\psi^j$ of the wave function $\psi$ can be seen as coupled wave functions on $\M$, but I will take the more common and natural view that they form a vector valued wave function or field, as I will explain now. Let's first fix a position $\x_0\in\M$. We define, for all $\psi:\M\times\D\to\C$, the function $\psi(\x_0):\D\to\C$, where
\begin{equation}
\label{eq:wave_D}
\psi(\x_0)(d^j) := \psi(\x_0,d^j).
\end{equation}

We notice the following:
\begin{enumerate}
	\item For any complex function $f:\D\to\C$, there is at least a wave function $\psi:\M\times\D\to\C$, so that $f=\psi(\x_0)$.
	\item For any two complex functions $f_1,f_2:\D\to\C$ and any two complex numbers $c_1,c_2\in\C$, $c_1f_1+c_2f_2$ is also a complex function on $\D$. For any two wave functions $\psi_j:\M\times\D\to\C$, so that $f_j=\psi_j(\x_0)$, $j\in\{1,2\}$, $c_1f_1+c_2f_2=c_1\psi_1(\x_0)+c_2\psi_2(\x_0)$.
\end{enumerate}

Therefore, at each $\x_0\in\M$, the collection of values of the form $\psi(\x_0)(d^j)$ form a complex vector space, which will be denoted by $\V\D$. A basis of $\V\D$ is given by the linear functions
\begin{equation}
\label{eq:vec_D_basis}
(\db_1,\ldots,\db_{\nD\D}),
\end{equation}
defined for all $d^j\in\D$ by
\begin{equation*}
\db_k(d^j) = \delta_k^j.
\end{equation*}
The dimension of the vector space $\V\D$ is $\dim\V\D=\nD\D$.

This is to say that
\begin{equation}
\label{eq:wave_tensor_prod_one}
\V{\M,\D} = \V{\M}\otimes\V\D,
\end{equation}
where $\otimes$ is the complex tensor product, and $\V{\M}$ is the space of scalar functions on $\M$.

This shows that the elements of $\V{\M,\D}$ are representable as wave functions on $\M$, or vector fields, valued in $\V\D$.

The elements of $\V\D$ are linear functions acting on the dual vector space $\V\D^*$ and conversely. Let $(\db^1,\ldots,\db^{\nD\D})$ be the basis on $\V\D^*$ which is dual to the basis \eqref{eq:vec_D_basis} of $\V\D$, {\ie} $\db^j(\db_k) = \delta^j_k$. By this, we have promoted the elements $d^j\in\D$ to a basis $(\db^1,\ldots,\db^{\nD\D})$ of $\V\D^*$.

We are now equipped to represent all wave functions of type $(\D_1\ldots,\D_\n)$ as in \eqref{eq:universal_wavefunction}, which are defined on $\M^\n \times \prod_{j=1}^\n\D_j$, as vector fields or vector-valued wave functions

\begin{equation}
\label{eq:universal_wavefunction_field}
\begin{split}
&\psi:\M^\n \to \otimes_{j=1}^\n\V{\D_j}  \\
&\psi(\x_1,\ldots,\x_\n)(\dn_{1},\ldots,\dn_{\n}) = \psi(\x_1,\ldots,\x_\n,\dn_{1},\ldots,\dn_{\n}). \\
\end{split}
\end{equation}

%------------------------------------------------------------%
\subsection{The Hilbert space}
\label{s:hilbert}

We will now recall Dirac's bra-ket notation. For $\x\in\M$, Dirac's delta distributions $\delta(\x)$ is a vector in $\V{\M}$, denoted by $\ket{\x}$. We also denote by $\ket{\dn_j}$ a basis element of the vector space $\V{\D}$. Then, a state vector of $\V{\M^\n,\D_1,\ldots,\D_\n}$ is denoted by
%\begin{widetext}
\begin{eqnarray}
\label{eq:psi_ket}
\ket{\psi} = \sum_{(\dn_{1},\ldots,\dn_{\n})\in\D_1\times\ldots\times\D_\n} 
\int_{\M^\n}\de\x_1\ldots\de\x_\n \nonumber \\
\times\psi(\x_1,\ldots,\x_\n,\dn_{1},\ldots,\dn_{\n}) \nonumber \\
\times\ket{\x_1,\ldots,\x_\n,\dn_{1},\ldots,\dn_{\n}},
\end{eqnarray}
%\end{widetext}
where
\begin{eqnarray}
\label{eq:psi_ket_prod}
\ket{\x_1,\ldots,\x_\n,\dn_{1},\ldots,\dn_{\n}} = \ket{\x_1}\otimes\ldots\otimes\ket{\x_\n} \nonumber \\
\otimes\ket{\dn_{1}}\otimes\ldots\otimes\ket{\dn_{\n}}.
\end{eqnarray}

We endow the vector space $\V\D$ with a Hermitian scalar product $\langle,\rangle_D$, uniquely characterized by the requirement that the basis \eqref{eq:vec_D_basis} of the vector space $\V\D$ generated by $\D$ is orthonormal. This scalar product induces a reciprocal Hermitian scalar product on $\V\D^*$, denoted by the same symbols $\langle,\rangle_D$ when no confusion can arise.
The scalar product of two state vectors $\psi_1:\M\to \D$ and $\psi_2:\M\to \D$ is 
\begin{equation}
\label{eq:scalar_prod_one}
\braket{\psi_1}{\psi_2} := \int_{\M}\langle\psi_1(\x),\psi_2(\x)\rangle_\D\de\x.
\end{equation}

The scalar product \eqref{eq:scalar_prod_one} extends to the many-particle wave functions, as the tensor product of the scalar products corresponding to each particle.
With its help, we can define the Hilbert space of square-integrable wave functions, and we can also work with a rigged Hilbert space, as usually done in \NRQM.

We have seen that the standard formulation of {\NRQM} already contains a way to reduce the wave function of a single particle from the configuration space $\M\times\D$ to a vector-valued wave function on the 3D space $\M$. But for $\n$ particles, this worked only to reduce it to $\M^\n$, as in equation \eqref{eq:universal_wavefunction_field}. In section \sref{s:st_rep} I will show how the reduction to $\M$ can be done for $\n$ particles as well.

%------------------------------------------------------------%
\subsection{The Fock space}
\label{s:fock}

For particles of the same type, if they are fermions, only the antisymmetric states are allowed, and if they are bosons, only the symmetric states. They are obtained by (anti)symmetrizing the tensor products of the Hilbert spaces for single particles, resulting in the Fock vector spaces corresponding to each type of particle.

The state space of $\n$ particles of the same type $\D$ is 
\begin{equation}
\label{eq:n_identical}
\F{\D}^\n{}_\pm := \Sym\pm\(\underbrace{\V{\M,\D}\otimes\ldots\otimes\V{\M,\D}}_{\tn{$\n$ times}}\),
\end{equation}
where the operator $\Sym+$ symmetrizes the tensor product, and $\Sym-$ anti-symmetrizes it. For \emph{bosons} $\Sym+$ is used, and for \emph{fermions}, $\Sym-$.

The \emph{Fock space} of particles of type $\D$ is
\begin{equation}
\label{eq:n_identical_fock}
\F{\D}{}_\pm := \bigoplus_{k=0}^\infty \F{\D}^k{}_\pm,
\end{equation}
where $\F{\D}^0{}_\pm\cong\C$ has only one dimension, being spanned by the field representation of the vacuum state. 

While the wave function for $\n$ particles can be seen as a vector-valued wave function defined on $\M^\n$, in the case of the Fock space it is a vector-valued wave function defined on
\begin{equation}
\label{eq:n_identical_config}
\bigsqcup_{k=0}^\infty \M^k,
\end{equation}
for both the fermionic and the bosonic cases, where $\M^0$ has only one point.

In general, the universal wave function is represented as a vector-valued wave function defined on
\begin{equation}
\label{eq:n_identical_config_total}
\underbrace{\(\bigsqcup_{k=0}^\infty \M^k\) \times \ldots \times\(\bigsqcup_{k=0}^\infty \M^k\)}_{\tn{$r$ times}},
\end{equation}
where $r$ is the number of distinct existing types of particles.

%------------------------------------------------------------%
\section{The 3D space representation of wave functions}
\label{s:st_rep}

The main result of this section is the following
\begin{theorem}
\label{thm:space_representation}
The space of many-particle wave functions defined on the configuration space (as in section \sref{s:def}) admits a representation as vector fields defined on the $3$D space $\M$.
\end{theorem}

In this section, I gradually develop the proof of this theorem, and the representation on which it is based, starting with two-particle separable states and continuing with increased generality.

%------------------------------------------------------------%
\subsection{Separable states}
\label{s:separable}

Let us start with a separable state vector

\begin{equation}
\label{eq:separable_state_two}
\ket{\psi} = \ket{\psi_1}\otimes\ket{\psi_2},
\end{equation}
where $\ket{\psi_1}\in\V{\M,\D_1}$ and $\ket{\psi_2}\in\V{\M,\D_2}$.

Recall from Sec. \sref{s:reduction_to_m_n} that $\psi_1$ and $\psi_2$ are vector-valued wave functions on $\M$. But $\psi$ is not, because even in the case when $\nD{\D_1}=\nD{\D_2}=1$, $\psi(\x_1,\x_2)=\psi_1(\x_1)\psi_2(\x_2)$ is defined on $\M\times\M$. We will see that it can be put in the form of a vector-valued wave function on $\M$.

A naive idea to do this is to think about $\psi$ as a two-component valued wave function, $(\psi_1,\psi_2)$. Such pairs are naturally vector fields from the direct sum bundle $\V{\M,\D_1}\oplus\V{\M,\D_2}$. This would be similar to the Pauli wave function for spin-$\frac 12$ particles, which can be seen as having two components, corresponding to the two possible values for the spin along the $z$ axis.

This may seem not very different from the two-components wave function, and in the case when other degrees of freedom are present, as a wave function with two components, each being a vector in $\D_1$ or $\D_2$. But, in order to see if this can work, we need to address some problems.

The first problem is that there are infinitely many ways to write $\psi$ as a tensor product of two one-particle states (which are representable on $\M$). More precisely, if there are two other wave functions $\psi'_1\in\V{\M,\D_1}$ and $\psi'_2\in\V{\M,\D_2}$ such that $\ket{\psi} = \ket{\psi'_1}\otimes\ket{\psi'_2}$, then there is a complex number $c\in\C_{\neq 0}=\C\setminus\{0\}$, so that 
\begin{equation}
\label{eq:non_uniqueness_separable_state_two}
\begin{split}
\psi'_1 &= c\psi_1 \\
\psi'_2 &= c^{-1}\psi'_2. \\
\end{split}
\end{equation}
If $\ket{\psi_1}\neq0$ and $\ket{\psi_2}\neq0$, then the number $c$ is unique. 
This means though that the representation of $\ket{\psi}$ as a pair $(\psi_1,\psi_2)$ is only ``almost'' unique. 

Fortunately, there is a way to remove this ambiguity, by factoring it out. 
We start with $\V{\M,\D_1}\oplus\V{\M,\D_2}$ to construct a representation of $ \ket{\psi'_1}\otimes\ket{\psi'_2}$ as a vector field over $\M$.

\begin{definition}
\label{def:btimes_equiv}
We first define a binary relation $\sim$ as
\begin{eqnarray}
\label{eq:btimes_equiv}
(\psi_1,\psi_2)\sim(\psi'_1,\psi'_2)\tn{ iff }\exists c\in\C_{\neq 0}\\
\tn{ so that }\psi'_1 = c\psi_1\tn{ and }\psi'_2 = c^{-1}\psi'_2.\nonumber
\end{eqnarray}
In this case, we say that $c$ \emph{ensures} the equivalence $(\psi_1,\psi_2)\sim(\psi'_1,\psi'_2)$. 
\end{definition}

\begin{proposition}
\label{thm:btimes_equiv}
The binary relation $\sim$ from Definition \ref{def:btimes_equiv} is an equivalence relation. 
\end{proposition}
\begin{proof}
Since $c=1$ means $(\psi_1,\psi_2)\sim(\psi_1,\psi_2)$, the relation $\sim$ is \emph{reflexive}.
The relation is \emph{symmetric}, because whenever $c$ ensures the equivalence $(\psi_1,\psi_2)\sim(\psi'_1,\psi'_2)$, $c^{-1}$ ensures $(\psi'_1,\psi'_2)\sim(\psi_1,\psi_2)$.
If there is a complex number $c'\neq 0$ ensuring $(\psi_1,\psi_2)\sim(\psi'_1,\psi'_2)$, and a complex number $c''\neq 0$ ensuring $(\psi'_1,\psi'_2)\sim(\psi''_1,\psi''_2)$, it follows that $c'c''$ ensures the equivalence $(\psi_1,\psi_2)\sim(\psi''_1,\psi''_2)$. Therefore, $\sim$ is \emph{transitive}. This proves that $\sim$ is an equivalence relation.
\end{proof}

\begin{definition}
\label{def:btimes}
We define the set
\begin{equation}
\label{eq:btimes}
\V{\M,\D_1} \btimess \V{\M,\D_2} := \(\V{\M,\D_1}\oplus\V{\M,\D_2}\)/\sim.
\end{equation}
Given a two-particle separable state $\ket{\psi_1}\otimes\ket{\psi_2}$, we denote the equivalence class defined by the equivalence relation $\sim$ by $\psi_1\btimes\psi_2:=\left[\psi_1,\psi_2\right]_\sim=\left[\psi_1\oplus\psi_2\right]_\sim$. 
\end{definition}

\begin{remark}
\label{rem:btimess_index}
Here, the index $s$ in the symbol $\btimess$ in $\V{\M,\D_1} \btimess \V{\M,\D_2}$ stands for ``separable'', and when I will generalize to the nonseparable case, the symbol $\btimes$ will be used instead. We will see that they stand for slightly different but related operations. The index $s$ is not present in the symbol $\btimes$ in $\psi_1\btimes\psi_2$ because this field will be essentially the same after the generalization.
\end{remark}

\begin{remark}
\label{rem:btimes_D}
Definitions \ref{def:btimes_equiv} and \ref{def:btimes} may be easier to understand if we consider first the scalar case, $\ket{\psi_1},\ket{\psi_2}\in\V{\M}$. To see how it works in general when $\ket{\psi_1}\in\V{\M,\D_1}$ and $\ket{\psi_2}\in\V{\M,\D_2}$, recall that $\ket{\psi_1}\otimes\ket{\psi_2}$ is expressed as $\psi_1(\x_1,v_1)\psi_2(\x_2,v_2)\in\C$, where $v_j\in\V{\D_j}$, $j\in\{1,2\}$, but also as $\psi_1(\x_1)\psi_2(\x_2)\in\V{\D_1}\otimes\V{\D_2}$. But since $\ket{\psi_1}\otimes\ket{\psi_2}$ is separable, $\psi_1(\x_1)\psi_2(\x_2)$ is also separable as an element of $\V{\D_1}\otimes\V{\D_2}$. Hence, the generalization of the relation $\sim$ from scalar wave functions to $\ket{\psi_1}\in\V{\M,\D_1}$ and $\ket{\psi_2}\in\V{\M,\D_2}$ is straightforward. Alternatively, we can think about the relation $\sim$ in terms of wave functions not on the configuration space $\M^\n$, but on $\M^\n \times \prod_{j=1}^\n\D_j$ (see section \sref{s:reduction_to_m_n}), reducing it again to the scalar wave functions case.
\end{remark}

\begin{remark}
\label{rem:btimes_transf}
Another way to define the equivalence relation from Definition \ref{def:btimes_equiv} is by using the multiplicative one-parameter group $\mathcal{G}_{\nD{\D_1},\nD{\D_2}}$ consisting of matrices of the form
\begin{equation}
\label{eq:btimes_transf}
\(
 \begin{matrix}
  c 1_{\nD{\D_1}} & 0 \\
  0 & c^{-1} 1_{\nD{\D_2}}
 \end{matrix}
\),
\end{equation}
where $c\in\C_{\neq 0}$. The elements of $\V{\M,\D_1} \btimess \V{\M,\D_2}$ are then the orbits of the action of this group on $\V{\M,\D_1}\oplus\V{\M,\D_2}$.
\end{remark}

\begin{remark}
The fields from $\V{\M,\D_1} \btimess \V{\M,\D_2}$ do not form a vector space.
This is easy to check. If $(\psi_1,\psi_2)$ and $(\psi_3,\psi_4)$ are vectors in $\V{\M,\D_1}\oplus\V{\M,\D_2}$, their sum is in general not equivalent to the sum of $(c_1\psi_1,c_1^{-1}\psi_2)$ and $(c_2\psi_3,c_2^{-1}\psi_4)$. In other words, the operation of addition of vectors in $\V{\M,\D_1}\oplus\V{\M,\D_2}$ does not ``survive'' the factorization by $\sim$.
This is to be expected, since $\V{\M,\D_1} \btimess \V{\M,\D_2}$ contains only representations of separable states, while the sum (superposition) of two separable state vectors is usually not a separable state vector. But here we used $\V{\M,\D_1}\oplus\V{\M,\D_2}$ only to represent pairs of vector-valued wave functions, and we are not interested in the vector space structure itself.
\end{remark}

We can now move from two-particle separable states to any $\n$.
\begin{definition}
\label{def:btimes_equiv_n}
We define a binary relation $\sim$ on $\V{\M,\D_1}\oplus\ldots\oplus\V{\M,\D_\n}$ by the following.

If $\n=1$, $(\psi)\sim(\psi')$ iff $\psi=\psi'$.

If $\n>1$,
\begin{equation}
\label{eq:btimes_equiv_n}
(\psi_1,\ldots,\psi_\n)\sim(\psi'_1,\ldots,\psi'_2)
\end{equation}
iff there is a linear transformation of $\V{\M,\D_1}\oplus\ldots\oplus\V{\M,\D_\n}$ of the form
\begin{equation}
\label{eq:btimes_transf_n}
T=
\(
 \begin{matrix}
  c_1 1_{\nD{\D_1}} & 0 & \ldots & 0\\
	\ldots & \ldots & \ldots & \ldots \\
  0 & \ldots & 0 & c_\n 1_{\nD{\D_\n}}
 \end{matrix}
\),
\end{equation}
where $c_1\ldots c_\n=1$,
such that $(\psi'_1,\ldots,\psi'_2) = T\(\psi_1,\ldots,\psi_\n\)$.
In this case, we say that $T$ \emph{ensures} the equivalence \eqref{eq:btimes_equiv_n}.

Let $\mathcal{G}_{\nD{\D_1},\ldots,\nD{\D_\n}}$ be the group of transformations of the form \eqref{eq:btimes_transf_n}.
The orbit
\begin{equation}
\label{eq:btimes_transf_n_orbit}
\mathcal{G}_{\nD{\D_1},\ldots,\nD{\D_\n}}(\psi_1,\ldots,\psi_\n)
\end{equation}
is the equivalence class $\left[\psi_1,\ldots,\psi_\n\right]_\sim$, and we denote it by $\psi_1\btimes\ldots\btimes\psi_\n$. Also, we define $\V{\M,\D_1} \btimess \ldots \btimess \V{\M,\D_\n}:=\(\V{\M,\D_1}\oplus\ldots\oplus\V{\M,\D_\n}\)/\sim$.
\end{definition}

\begin{proposition}
\label{thm:btimes_equiv_n}
The binary relation $\sim$ from Definition \ref{def:btimes_equiv_n} is an equivalence relation. 
\end{proposition}
\begin{proof}
The proof is similar to that of the Proposition \ref{thm:btimes_equiv}.

\emph{Reflexivity} follows by taking $T=1_{\nD{\D_1}+\ldots+\nD{\D_\n}}$ (the identity).

If $(\psi_1,\ldots,\psi_\n)\sim(\psi'_1,\ldots,\psi'_\n)$ is ensured by $T$, then $(\psi'_1,\ldots,\psi'_\n)\sim(\psi_1,\ldots,\psi_\n)$ is ensured by $T^{-1}$, so $\sim$ is \emph{symmetric}.

If there are two matrices $T',T''$, as in equation \eqref{eq:btimes_transf_n}, so that $T'$ ensures $(\psi_1,\ldots,\psi_\n)\sim(\psi'_1,\ldots,\psi'_\n)$, and $T''$ ensures $(\psi'_1,\ldots,\psi'_\n)\sim(\psi''_1,\ldots,\psi''_\n)$, it follows that $T''T'$ is from $\mathcal{G}_{\nD{\D_1},\ldots,\nD{\D_\n}}$, and ensures the equivalence $(\psi_1,\ldots,\psi_\n)\sim(\psi''_1,\ldots,\psi''_n)$. Therefore, $\sim$ is \emph{transitive}. This proves that $\sim$ is an equivalence relation.
\end{proof}

\begin{remark}We could have tried to represent tensor product states as vector fields from the tensor products of vector bundles, but this doesn't work for some basic reasons. First, the tensor products of vector spaces of vector fields from two bundles is much larger than the vector space of vector fields from the tensor product of those vector bundles. Second, which is the main reason I didn't use it, if we would use this for the case when one of the vector fields vanishes at regions where the other one does not vanish, their tensor product would vanish, which would lose information contained in the original vector fields.
\end{remark}

Definition \ref{def:btimes_equiv_n} still misses something, which will be added in the following subsection.

%------------------------------------------------------------%
\subsection{Local separability}
\label{s:local_separability}

A question that arises is the following: given two regions $A,B\subseteq\M$, and a field representation $\widetilde{\Psi}=\left[\psi_1,\ldots,\psi_\n\right]_\sim$, can we recover $\widetilde{\Psi}|_{A\cup B}$ if we know $\widetilde{\Psi}|_A$ and $\widetilde{\Psi}|_B$? On the one hand, this seems impossible for the wave function on the configuration space, for the simple reason that it is defined on $\M^\n$, and not on $\M$. And trying to do this with field operators in quantum field theory, even in its ``local'' version, doesn't work, because even the vacuum is locally non-separable \cite{SEP-Healey2016-physics-holism}.

To qualify as fields on the $3$D space, the representations we give to the many-particle states have to be defined on $\M$ and valued in some fiber, so they have to be vector fields from a bundle over $\M$. In other words, a field on $\M$ should be defined by its values in all $\x\in\M$. In particular, this would ensure that we can recover $\widetilde{\Psi}|_{A\cup B}$ from its restrictions $\widetilde{\Psi}|_A$ and $\widetilde{\Psi}|_B$.

At first sight, the equivalence classes from Definition \ref{def:btimes_equiv_n} do not seem to work like fields:

\begin{problem}
\label{problem:local_symmetry_problem}
Consider a partition of $\M$, $\M=\bigsqcup_k M_k$, and the restrictions $\widetilde{\Psi}|_{M_k}$ of a field $\widetilde{\Psi}=\left[\psi_1,\ldots,\psi_\n\right]_\sim$ on each $M_k$ (in particular, the partition can be $\M=\bigsqcup_{\x\in\M} \{\x\}$). Then, just by knowing the restrictions $\widetilde{\Psi}|_{M_k}$, it is impossible to recover the original field $\widetilde{\Psi}$.
\end{problem}

On the one hand, since $\ket{\psi_1}\otimes\ket{\psi_2}=(c\ket{\psi_1})\otimes(c^{-1}\ket{\psi_2})$ we had to get rid of the differences between $(\psi_1,\psi_2)$ and $(c\psi_1,c^{-1}\psi_2)$, and this was achieved in Definition \ref{def:btimes} and generalized in Definition \ref{def:btimes_equiv_n}. On the other hand, the representation of states has to be defined pointwisely, to qualify as a field, so it seems that we should keep somehow the difference between $(\psi_1,\psi_2)$ and $(c\psi_1,c^{-1}\psi_2)$. To achieve both of these objectives, we do what we do in gauge theory -- we treat the equivalence classes $\left[\psi_1,\psi_2\right]_\sim$ as the true physical fields, and we keep $(\psi_1,\psi_2)$ as a representation of our field in a particular choice of gauge.

\begin{solution}{Problem}{\ref{problem:local_symmetry_problem}}
\label{solution:local_symmetry_problem}
Now I will explain this in detail, with the help of some standard notions of fiber bundles \cite{NashSen1983TopologyGeometryPhysicists,Nakahara2003GeometryTopologyPhysics,Bleecker2005GaugeTheoryVariationalPP,ROWells2008ComplexManifolds}. We know that to a vector bundle with typical fiber $V$, we can associate a \emph{principal bundle}, with \emph{structure group} $\GL(V)$, the group of linear transformations of $V$. In our case, the vector bundles $\V{\M,\D_j}$ have as typical fibers vector spaces $\V{\D_j}$, so the structure groups for each of these bundles are $\GL(\V{\D_j})$. Since the topology of $\M$ is that of $\R^3$, let's take the associated principal bundle to each $\V{\M,\D_j}$ to simply be the trivial bundle $\M\times\GL(\V{\D_j})$. A \emph{gauge} of the principal bundle $\M\times\GL(\V{\D_j})$ is a \emph{frame field} of $\V{\M,\D_j}$, so if we are given the representation of $\psi\in\Gamma\(\V{\M,\D_j}\)$ in components, we also need to be given the frame field.
In our representation, the structure group is a subgroup of the group $\GL(\V{\D_1}\oplus\ldots\oplus\V{\D_\n})$.
Since here we are not concerned with the linear transformations of the bundles $\V{\D_j}$, we will consider for simplicity that the frame field is fixed everywhere.
We are then free to focus on the action of the group $\mathcal{G}_{\nD{\D_1},\ldots,\nD{\D_\n}}$ from Definition \ref{def:btimes_equiv_n}, since Problem \ref{problem:local_symmetry_problem} comes only from the equivalence \eqref{eq:btimes_transf_n}.

The group $\mathcal{G}_{\nD{\D_1},\ldots,\nD{\D_\n}}$ from Definition \ref{def:btimes_equiv_n} is isomorphic to the commutative group $\mathcal{G}_\n:=\C_{\neq 0}^{\n-1}:=\underbrace{\C_{\neq 0}\times\ldots\times\C_{\neq 0}}_{\tn{${\n-1}$ times}}$, where $\C_{\neq 0}:=\C\setminus\{0\}$ is the multiplicative group of complex numbers. Hence, ignoring for simplicity the linear transformations of each one-particle vector bundle, the principal bundle associated to our representation is $\M\times\mathcal{G}_\n$. 
Each vector field $g(\x)=(c_1(\x),\ldots,c_{\n-1}(\x))$ of the principal bundle $\M\times\mathcal{G}_\n$ acts on the vector bundle $\V{\M,\D_1}\oplus\ldots\oplus\V{\M,\D_\n}$ by a transformation
\begin{equation}
\label{eq:btimes_transf_local_n}
T(g)(\x):=
{\tiny\(
 \begin{matrix}
  c_1(\x) 1_{\nD{\D_1}} & \ldots & 0 & 0\\
	\ldots & \ldots & \ldots & \ldots \\
  0 & \ldots & c_{\n-1}(\x) 1_{\nD{\D_{\n-1}}}& 0 \\
  0 & \ldots & 0 & c^{-1}_1\ldots c^{-1}_{\n-1} 1_{\nD{\D_\n}} \\
 \end{matrix}
\)}.
\end{equation}

To obtain the desired bundle whose vector fields represent the many-particle wave functions, we apply now the standard associated bundle construction.
We take the product bundle $\(\M\times\mathcal{G}_\n\)\times\(\V{\M,\D_1}\oplus\ldots\oplus\V{\M,\D_\n}\)$, and let the group $\mathcal{G}_\n$ act on it by the right action
%\begin{widetext}
\begin{eqnarray}
\label{eq:group_action_btimes}
\(p(\x),(\psi_1(\x),\ldots,\psi_\n(\x))\) \qquad\qquad\qquad\qquad\qquad\nonumber\\
:= \(p(\x) g(\x),T(g^{-1})(\x)(\psi_1(\x),\ldots,\psi_\n(\x))\).
\end{eqnarray}
%\end{widetext}

The orbits of this action, $\left[p(\x),(\psi_1(\x),\ldots,\psi_\n(\x))\right]_\sim$, form an associated bundle 
\begin{multline}
\label{eq:btimes_assoc_bundle}
\V{\M,\D_1 \btimess \ldots \btimess \D_\n} := \(\M\times\mathcal{G}_\n\)\times_{T}\(\V{\M,\D_1}\oplus\ldots\oplus\V{\M,\D_\n}\) \\
= \(\(\M\times\mathcal{G}_\n\)\times\(\V{\M,\D_1}\oplus\ldots\oplus\V{\M,\D_\n}\)\)/\mathcal{G}_\n,
\end{multline}
whose base manifold is $\M$. As a fiber, the gauge group plays in the principal bundle the role of a \emph{torsor}, {\ie} we forget its group structure and we keep it as a homogeneous space, as in the case of the frame bundle. This allows us to treat the degrees of freedom of the associated bundle as unphysical, but at the same time in a local separable way.

By this, we have seen that the construction obtained in Definition \ref{def:btimes_equiv_n} is in fact a fiber bundle, and we are justified to consider $\psi_1\btimes\ldots\btimes\psi_\n$ as its vector fields, hence as fields defined on the $3$D space $\M$.

Returning to Problem \ref{problem:local_symmetry_problem}, we see that the right construction was given here, rather than in Definition \ref{def:btimes_equiv_n}, and the restrictions $\widetilde{\Psi}|_{M_k}$ contain not only the equivalence classes, but also the gauge in which they are expressed. This removes the ambiguity and allows to recover $\widetilde{\Psi}$ on the entire $3$D space $\M$ from its restrictions $\widetilde{\Psi}|_{M_k}$. In particular, we can recover $\widetilde{\Psi}$ on the entire $3$D space $\M$ by knowing its values $\widetilde{\Psi}(\x)$. 
\end{solution}

\begin{remark}
\label{rem:btimes_bundle_operations}
One word of caution is in order. The bundle \eqref{eq:btimes_assoc_bundle} is a vector bundle, and carries natural mathematical operations with vector fields. But they do not correspond to physical operations with many-particle wave functions. Once we added the gauge symmetry given by the group $\mathcal{G}_\n$, we can, in principle, add two separable state vectors $(\psi_1,\psi_2)$ and $(\psi_3,\psi_4)$, and obtain another separable state vector $(\psi_1+\psi_3,\psi_2+\psi_4)$. This would break our gauge symmetry, or the equivalence relation from Definition \ref{def:btimes}, since given other representatives $(\psi'_1,\psi'_2)\sim(\psi_1,\psi_2)$ and $(\psi'_3,\psi'_4)\sim(\psi_3,\psi_4)$, in general $(\psi'_1+\psi'_3,\psi'_2+\psi'_4)\nsim(\psi_1+\psi_3,\psi_2+\psi_4)$. The reason why this is not a valid operation is that the Hamiltonian is additive, but only when acting on the equivalence class $\psi_1\btimes\psi_2$, not on its representatives $(\psi_1,\psi_2)$.
And this works for us, because superpositions of separable states are in general non-separable.
\end{remark}

\begin{remark}
\label{rem:btimes_gauge}
Our representation of separable states was obtained \emph{only by using the following operations with vector bundles}: the direct sum, which is a local operation, and an equivalence relation, which is a global operation akin to global symmetries, in the sense that the transformations \eqref{eq:btimes_transf_n} ensuring the equivalence relation are constant over $\M$. In particular, this is similar to changing the phase of a wave function over $\M$ (which is in fact a global gauge transformation, {\eg} for the electron wave function it corresponds to the global $\U(1)$ symmetry of electromagnetism), or to other global symmetries.
We can promote the global transformations \eqref{eq:btimes_transf_n} of $\V{\M,\D_1}\oplus\ldots\oplus\V{\M,\D_\n}$ to local transformations $T$, making the coefficients $c_1,\ldots,c_\n$ depend on the position $\x\in\M$, provided that we keep track of the gauge, $\left[\psi_1,\ldots,\psi_\n\right]_\sim\neq\left[T(\psi_1,\ldots,\psi_\n\right)]_\sim$ in general, if $T$ is a local transformation. But this is not the right way to do it, the right way to do it involves gauge transformations like in equation \eqref{eq:group_action_btimes}.
A formulation allowing \emph{local gauge transformations} of the form \eqref{eq:btimes_transf_local_n} implies, for the differential operators, a new gauge connection. If we want to keep the correspondence with the separable states, this connection has to be flat, its only use being to compensate for the local gauge transformations. But the main point is that the resulting fields are defined pointwisely on the $3$D space $\M$.
Note that, in the presence of electromagnetic interactions, additional transformations appear. They have similar matrix form as \eqref{eq:btimes_transf_n}, but the coefficients $c_j$ are phase factors, the phase change is proportional with the electric charge of each particle type, and the product of all coefficients $c_j$ doesn't have to be $1$ (unless the total charge is $0$). The associated connection, corresponding to the electromagnetic potentials, is not flat. In Definition \ref{def:btimes_equiv_n} I could have avoided the condition that the product of all coefficients $c_j$ is $1$, obtaining a unified treatment of the gauge transformation introduced here to obtain the representation, and those of electromagnetism. I prefer not to do it here because it would complicate the exposition beyond the intended scope of this article.
\end{remark}

So far we have made some progress in representing tensor products of one-particle states as fields defined on the $3$D space.
This representation captures the intuition many researchers have, that somehow separable states are indeed separable, {\ie} they can be seen as separate wave functions on the $3$D space. But we have seen that this construction was not straightforward, because $\n$ one-particle wave functions contain more information than their product state. So we had to factor out this redundancy, which led to further complications, since factoring it out makes the fields to lose local separability. To restore it, we still take the separable states as consisting of $\n$ one-particle wave functions, and the redundancy is interpreted as a new symmetry, like local gauge symmetry (but without associated interactions).

To prove the full equivalence with the tensor products, {\ie}, to also include nonseparable states, we still need some work, which is done in the following subsection. Unfortunately, this again is not straightforward.

%------------------------------------------------------------%
\subsection{More about separable fields}
\label{s:more_separable}

I will now establish some properties and operations with the fields introduced in \sref{s:separable} to represent separable states.

%------------------------------------------------------------%
\subsubsection{One dimensional vector space of fields}
\label{s:one_dim_hilbert}

Factoring by the equivalence relation $\sim$ does not preserve all the vector space operations (Remark \ref{rem:btimes_bundle_operations}). But it is not necessary to preserve them, since the sum of separable states in general is not separable. Moreover, the equivalence relation $\sim$ does commute with some of the vector space operations, just the way we need, as we shall see.

\begin{definition}
\label{def:btimes_scalar_mult}
On the set $\V{\M,\D_1}\btimess\ldots\btimes\V{\M,\D_\n}$, we define the scalar multiplication with a complex number $c\in\C$ by 
\begin{equation}
\label{eq:btimes_scalar_mult}
c\psi_1\btimes\ldots\btimes\psi_\n := (c\psi_1)\btimes\ldots\btimes\psi_\n.
\end{equation}
Multiplying by $c\neq 1$ changes the equivalence class $\psi_1\btimes\ldots\btimes\psi_\n$, so this operation is well defined. 
We say that $\psi_1\btimes\ldots\btimes\psi_\n$ and $c\psi_1\btimes\ldots\btimes\psi_\n$ are \emph{collinear}.
Let $\Span\(\psi_1\btimes\ldots\btimes\psi_\n\)$ be the one-dimensional vector space spanned by $\psi_1\btimes\ldots\btimes\psi_\n$ by multiplications with scalars.
\end{definition}

Definition \ref{def:btimes_equiv_n} allows us to move the scalar $c\in\C_{\neq0}$ among the factors, $(c\psi_1)\btimes\psi_2\btimes\ldots\btimes\psi_\n\sim\psi_1\btimes(c\psi_2)\btimes\ldots\btimes\psi_\n\sim\ldots\sim\psi_1\btimes\psi_2\btimes\ldots\btimes(c\psi_\n)$. 

\begin{definition}
\label{def:btimes_collinear_addition}
We can even add fields, provided that their equivalence classes are collinear.
\begin{equation}
\label{eq:btimes_collinear_addition}
c_1\psi_1\btimes\ldots\btimes\psi_\n + c_2\psi_1\btimes\ldots\btimes\psi_\n := (c_1+c_2)\psi_1\btimes\ldots\btimes\psi_\n.
\end{equation}
\end{definition}

These operations turn the set of all collinear fields into a one-dimensional vector space. This very simple observation will turn out to be very useful in the following.

%------------------------------------------------------------%
\subsubsection{Recursivity and associativity}
\label{s:recursivity_associativity}

\begin{remark}
\label{rem:recursivity}
Due to the operations defined in \sref{s:one_dim_hilbert}, we can apply Definition \ref{def:btimes_equiv_n} \emph{recursively}, allowing $\psi_j$ to be not only one-particle wave functions, but also many-particle separable states. The reason is that, in Definition \ref{def:btimes_equiv_n}, only direct sums and scalar multiplications are used.
This makes possible to talk about associativity, which I will extend now to the proposed field representation.
\end{remark}

\begin{proposition}
\label{thm:btimes_equiv_assoc}
Let $\psi_1,\psi_2,\psi_3$ be one-particle wave functions or many-particle separable states. Then, $(\psi_1\btimes\psi_2)\btimes\psi_3=\ket{\psi_1}\btimes(\ket{\psi_2}\btimes\psi_3)=\psi_1\btimes\psi_2\btimes\psi_3$.
\end{proposition}
\begin{proof}
Consider any two complex numbers $c_1,c_2\in\C_{\neq 0}$.

Then, for the identity
$(\psi_1\btimes\psi_2)\btimes\psi_3=\psi_1\btimes\psi_2\btimes\psi_3$,
\begin{align*}
((\psi_1,\psi_2),\psi_3) & \sim (c_2(c_1\psi_1,c_1^{-1}\psi_2),c_2^{-1}\psi_3) \\
& = (c_2c_1\psi_1,c_1^{-1}\psi_2,c_2^{-1}\psi_3). \\
\end{align*}

If we make the notation $c_1'=c_1c_2$, $c_2'=c_1^{-1}$, and $c_3'=c_2^{-1}$, we obtain $c_1'c_2'c_3'=1$. 
Then, we can also solve for $c_1=c_2'{}^{-1}$, $c_2=c_3'{}^{-1}$, and $c_3=c_1'{}^{-1}$.
Then, the pairs $(c_1,c_2)$ are in one-to-one correspondence with triples $(c_1',c_2',c_3')$ so that $c_1'c_2'c_3'=1$, which proves that 
\begin{equation*}
((\psi_1,\psi_2),\psi_3) \sim (\psi_1,\psi_2,\psi_3),
\end{equation*}
hence $(\psi_1\btimes\psi_2)\btimes\psi_3=\psi_1\btimes\psi_2\btimes\psi_3$.

The identity $\psi_1\btimes(\psi_2\btimes\psi_3)=\psi_1\btimes\psi_2\btimes\psi_3$ follows similarly.
\end{proof}

\begin{remark}
Since the operation $\btimes$ is associative, it is therefore convenient to drop the brackets.
This property extends immediately to a general number of factors, because it applies to one-particle as well as separable many particle states as well, and takes us closer to the relation with the tensor products of quantum states.
\end{remark}

%------------------------------------------------------------%
\subsection{Nonseparable states}
\label{s:nonseparable}

Recall that the quotient set $\V{\M,\D_1} \btimess \ldots \btimess \V{\M,\D_\n}$ defined in \eqref{eq:btimes_equiv} is not a vector space. 
In \sref{s:one_dim_hilbert} we have seen that we can identify collinear fields, and they form one-dimensional vector spaces. But the representation of the nonseparable states has to be obtained, as linear combinations of separable states. It is predictable by now that, in order to achieve this, we can simply build sums of fields representing separable states. While it is not as straightforward, it is easy.

The main problems to be solved by our construction are:
\begin{enumerate}
	\item 
We cannot simply take direct sums of all such bundles, because they may be \emph{redundant}. The cause of this redundancy is that not all linear combinations of separable state vectors are not separable, for example $\ket{\psi_1}\otimes\ket{\psi_2}+\ket{\psi_1}\otimes\ket{\psi_2'} = \ket{\psi_1}\otimes\(\ket{\psi_2}+\ket{\psi_2'}\)$ is separable.
	\item 
The resulting operation of addition used to represent nonseparable states as superpositions of separable states has to be \emph{commutative}.
The direct sum of vector bundles is commutative, in the sense that given two vector bundles $E_1$ and $E_2$, $E_1\oplus E_2$ and $E_2\oplus E_1$ are isomorphic, but the direct sum of two particular vector fields from these bundles is not commutative. This is in fact already clear when we take the direct sum of vectors, since $(v_1,v_2)\neq(v_2,v_1)$, so $v_1\oplus v_2\neq v_2\oplus v_1$. The way commutativity works in a direct sum vector space is rather $(v_1,0)+(0,v_2) = (0,v_2)+(v_1,0)$, but these vectors are from the direct sum, not from the original vector spaces taking part in the sum. The fact that $(v_1,v_2)\neq(v_2,v_1)$ allowed us to define the operation $\btimes$ in the first place, but for addition we need to be careful to ensure commutativity.
\end{enumerate}

One way to avoid these problems, which is not used here, is to take all possible direct sums of bundles representing separable states, then identify what kinds of fields represent the same quantum states, and then factor out the redundancy, and do this in a way to ensure commutativity of addition. 

The method used in the following avoids the redundancy from the very beginning and obtains commutativity automatically.

\begin{remark}
\label{rem:tensor_product_basis}
Recall that given the vector spaces $V_1,\ldots,V_n$, and a basis $\(e_1^{(k)},\ldots,e_{\dim V_k}^{(k)}\)$ for each $V_k$, then $\(e_{j_1}^{(1)}\otimes\ldots\otimes e_{j_n}^{(n)}\)_{j_1=1\tn{ to }\dim V_1,\ldots,j_n=1\tn{ to }\dim V_n}$ is a basis of the tensor product $V_1\otimes\ldots\otimes V_n$.
\end{remark}

Now, let us build the space of all possible fields representing many-particle states, where the types of distinct particles are $\D_1,\ldots,\D_\n$.
The steps of the construction are as following:
\begin{enumerate}
	\item
	Let $\V\M$ be the vector space of scalar functions on $\M$, and $\(\xi_{\alpha}\)$ a basis of $\V\M$, indexed by $\alpha$. 
	Let $\(\db_{(j)}^{k}\)_k=\(\db_{(j)}^{1},\ldots,\db_{(j)}^{\nD{\D_j}}\)$ be a basis of $\V{\D_j}$, indexed by $k\in\{1,\ldots,\nD{\D_j}\}$, for each type of particle $\D_j$.
	Then, if  $\xi_{\alpha k}^{(j)} := \xi_{\alpha}\db_{(j)}^{k}$,
	\begin{equation}
	\label{eq:basis_nonseparable_field}
	\(\xi_{\alpha k}^{(j)}\)
	\end{equation}
	is a basis of $\V{\M,\D_j} := \V\M\otimes\V{\D_j}$, indexed by $\alpha$ and $k$, for each type of particle $\D_j$.
	\item
	For each $\n\geq 1$, construct the fields representing $\n$ particles, of the form 
	\begin{equation}
	\label{eq:basis_many_particle}
	\xi_{\alpha k_1}^{(j)}\btimes\ldots\btimes\xi_{\alpha k_\n}^{(j)} := \left[\xi_{\alpha k_1}^{(j)},\ldots,\xi_{\alpha k_\n}^{(j)}\right]_\sim,
	\end{equation}
	as in Definition \ref{def:btimes_equiv_n}, for all possible elements of the bases and all types of particles.
	\item
	Form the direct sum of all one-dimensional vector spaces spanned by vectors of the form $\xi_{\alpha k_1}^{(j)}\btimes\ldots\btimes\xi_{\alpha k_\n}^{(j)}$.
\end{enumerate}

\begin{remark}
\label{rem:nonseparable}
This construction relies only on direct sums of bundles as in the case of separable states. Since the fields representing separable states are defined on the $3$D space, the direct sums taken here are also defined on the $3$D space. Due to Remark \ref{rem:tensor_product_basis}, the construction provides a faithful representation of the space of many-particle wave functions of all types from the possible types $\D_1,\ldots,\D_\n$, as fields on the $3$D space. The discussion in \sref{s:local_separability} applies in this case too.
It is clear that the fiber bundle defined like this is very complicated and the fibers are infinite dimensional. This is to be expected, because otherwise we could not represent the many-particle wave functions, normally defined on the configuration space, as fields on just a $3$-dimensional space.
\end{remark}

\begin{remark}
\label{rem:nonseparable_commutativity}
The commutativity of addition follows now automatically, since we identify the vector bundles representing separable states as being components of the direct sum of all one-dimensional vector spaces spanned by the vector $\xi_{\alpha k_1}^{(j)}\btimes\ldots\btimes\xi_{\alpha k_\n}^{(j)}$. As explained already, it is similar to the difference between $v_1\oplus v_2\neq v_2\oplus v_1$ (non-commutativity) and $(v_1,0)+(0,v_2)=(0,v_2)+(v_1,0)$ (commutativity).
\end{remark}

How do we represent a generic $\n$-particle wave function? In particular, if $\psi_1,\ldots,\psi_n$ are $\n$ one-particle wave functions, how do we represent the field $\left[\psi_1,\ldots,\psi_n\right]_\sim$?
Let $\D_{j}$ be the type of each $\psi_j$. We express each $\psi_j$ in the basis $\(\xi_{\alpha k_{(j)}}^{(j)}\)$ of its space of functions $\V{\M,\D_{j}}$,
\begin{equation}
\label{eq:basis_expression_psi}
\psi_j=\sum_{k_{(j)}} c^{(j)}_{\alpha_{k_{(j)}}} \xi_{\alpha k_{(j)}}^{(j)},
\end{equation}
where the coefficients $c^{(j)}_{\alpha_{k_{(j)}}}$ are complex numbers.
Then, we \emph{define}
\begin{equation}
\label{eq:basis_expression_psi_btimes}
{\psi_1}\btimes\ldots\btimes{\psi_\n} := \sum_{k_{(1)}}\ldots\sum_{k_{(\n)}} c^{(1)}_{\alpha_{k_{(1)}}}\ldots c^{(\n)}_{\alpha_{k_{(\n)}}} \xi_{\alpha k_{(1)}}^{(1)}\btimes\ldots\btimes\xi_{\alpha k_\n}^{(\n)}.
\end{equation}

\begin{remark}
\label{rem:justification_basis_expression_psi_btimes}
Let me emphasize what I did here. I first defined the operation $\btimes$ as in Definition \ref{def:btimes_equiv_n}, but only on the elements of the basis, which was chosen from the beginning. Then, I extended the operation $\btimes$ to more general separable states, which are linear combinations of the fields corresponding to tensor products of elements of the bases. Alternatively, I could have defined the operation $\btimes$ for all separable states, then impose equation \eqref{eq:basis_expression_psi_btimes} as an equivalence relation, and then take the equivalence classes. The result would have been the same, regardless if we take equation \eqref{eq:basis_expression_psi_btimes} as an identity, or as an equivalence relation. The procedure I chose is in fact similar to the multiplication with a scalar from Definition \ref{eq:btimes_scalar_mult}. In both cases, I took advantage of the freedom that some operations are not defined, and I defined them to connect some fields which were previously independent.% As a result, we can identify one-dimensional vector spaces of fields as in Definition \ref{def:btimes_scalar_mult} as subspaces of 
\end{remark}

\begin{proposition}
\label{thm:btimes_distributive}
The operation $\btimes$ is distributive over the addition,
\begin{align}
\label{eq:btimes_distributive}
{\psi_1} \btimes \({\psi_2} + {\psi_3}\) &= {\psi_1} \btimes {\psi_2} + {\psi_1} \btimes {\psi_3} \\
\({\psi_2} + {\psi_3}\) \btimes {\psi_1} &= {\psi_2} \btimes {\psi_1} + {\psi_3} \btimes {\psi_1},
\end{align} 
where ${\psi_1}\in\V{\M,\D_1}$ and ${\psi_2},{\psi_3}\in\V{\M,\D_2}$.
\end{proposition}
\begin{proof}
This follows from equation \eqref{eq:basis_expression_psi_btimes}.
\end{proof}

\begin{remark}
\label{rem:nonseparable_other_basis}
This construction relies on choosing a particular basis for each one-particle Hilbert space, but the identity \eqref{eq:basis_expression_psi_btimes} allows us to change the basis. This makes the construction independent on the basis we choose.
\end{remark}

\begin{definition}
\label{def:btimes_space}
The fields of the form
\begin{equation}
\label{eq:btimes_vector_spaces}
\sum_{k_{(1)}}\ldots\sum_{k_{(\n)}} c_{\alpha_{k_{(1)}}\ldots \alpha_{k_{(\n)}}} \xi_{\alpha k_{(1)}}^{(1)}\btimes\ldots\btimes\xi_{\alpha k_\n}^{(\n)}
\end{equation}
with $c_{\alpha_{k_{(1)}}\ldots \alpha_{k_{(\n)}}}\in\C$, form a vector space, which we denote by $\V{\M,\D_{1}} \btimes \ldots \btimes \V{\M,\D_{\n}}$.
\end{definition}

\begin{definition}
\label{def:btimes_space_rep_partial}
We denote by
\begin{equation}
\label{eq:btimes_space_rep_partial}
\begin{cases}
\rep:\V{\M,\D_{1}} \otimes \ldots \otimes \V{\M,\D_{\n}} \to \V{\M,\D_{1}} \btimes \ldots \btimes \V{\M,\D_{\n}} \\
\rep(\ket{\psi_1}\otimes\ldots\otimes\ket{\psi_n}) = {\psi_1}\btimes\ldots\btimes{\psi_n} \\
\end{cases}
\end{equation}
(extended by linearity)
the faithful representation of the tensor product space $\V{\M,\D_{1}} \otimes \ldots \otimes \V{\M,\D_{\n}}$ on the vector space $\V{\M,\D_{1}} \btimes \ldots \btimes \V{\M,\D_{\n}}$ from Definition \ref{def:btimes_space}.
\end{definition}

The last step is to extend the representation \eqref{eq:btimes_space_rep_partial} to the direct sum $\mathcal{V}$ of all spaces of the form $\V{\M,\D_{1}} \otimes \ldots \otimes \V{\M,\D_{\n}}$, which is almost immediate. ``Almost'', because we also have to include in the direct sum the one-particle spaces, which are just the spaces $\V{\M,\D_{k}}$, and the vacuum state space.

Definition \ref{def:btimes_space} doesn't apply to the vacuum state, so its representation needs to be defined separately. It has to be the same for all types of particles. It generates a one-dimensional vector space, and it is independent on the type of particles. So the field representing it has to be a scalar field, with no relation to the internal spaces $\D_k$. It has to be invariant to isometries of the $3$D space $\M$, so it has to be constant. So we take the field representation of the vacuum state as being identically $1$. But in this case it is not square-integrable. Fortunately, the Hermitian scalar product on this space will be induced by that of the usual vacuum state space, in equation \eqref{eq:hermitian_scalar_prod_fields}. This makes sense, because the vacuum field space contains only constant functions, while the usual square integral doesn't apply.

\begin{definition}
\label{def:btimes_space_rep}
Let $\mathcal{V}$ be the direct sum of all state spaces of the form $\V{\M,\D_{1}} \otimes \ldots \otimes \V{\M,\D_{\n}}$, and $\widetilde{\mathcal{V}}$ the direct sum of all field spaces of the form $\V{\M,\D_{1}} \btimes \ldots \btimes \V{\M,\D_{\n}}$ (including the case of one particle and the vacuum). Then, we extend the representation from Definition \ref{def:btimes_space_rep_partial} by
\begin{equation}
\label{eq:btimes_space_rep}
\begin{cases}
\rep:\mathcal{V} \to \widetilde{\mathcal{V}} \\
\rep(\sum_j\ket{\Psi_j}) = \sum_j\rep(\ket{\Psi_j}), \\
\end{cases}
\end{equation}
where each $\rep\(\ket{\Psi_j}\)\stackrel{\tn{notation}}{=}\ket{\widetilde{\Psi_j}}$ belongs to a field space of the form $\V{\M,\D_{1}} \btimes \ldots \btimes \V{\M,\D_{\n}}$, including the one particle spaces, or is the vacuum.
\end{definition}

Again, the fields from this representation are defined on the $3$D space $\M$.

We also need to define the Hermitian scalar product on $\widetilde{\hilbert}$. To do this, we simply use the representation \eqref{eq:btimes_space_rep}. That is, for $\ket{\widetilde{\Psi}_1},\ket{\widetilde{\Psi}_2}\in\widetilde{\hilbert}$, 
\begin{equation}
\label{eq:hermitian_scalar_prod_fields}
\braket{\widetilde{\Psi}_1}{\widetilde{\Psi_2}} := \braket{\rep^{-1}(\widetilde{\Psi}_1)}{\rep^{-1}(\widetilde{\Psi}_2)}.
\end{equation}

Note that the field vector spaces $\V{\M,\D_{j}}$ can be replaced in the representation by any subspaces, in particular by the Hilbert spaces of square integrable fields, which we will denote by $\hilbert_{\D_{j}}$. In this case, let us denote by $\hilbert$ the direct sum of all spaces of the form $\hilbert_{\D_{1}} \otimes \ldots \otimes \hilbert_{\D_{\n}}$, and by $\widetilde{\hilbert}$ the direct sum of all field spaces of the form $\hilbert_{\D_{1}} \btimes \ldots \btimes \hilbert_{\D_{\n}}$ (including the case of one particle and the vacuum).

Similarly to equation \eqref{eq:hermitian_scalar_prod_fields}, the isomorphism \eqref{eq:btimes_space_rep} allows us to associate, to any operator $\hat{\mathcal{A}}$ on $\hilbert$ (or $\mathcal{V}$), an operator $\rep(\hat{\mathcal{A}})$ on $\widetilde{\hilbert}$ (or $\widetilde{\mathcal{V}}$), by
\begin{equation}
\label{eq:operators_fields}
\rep(\hat{\mathcal{A}})\ket{\widetilde{\Psi}_1} := \rep\(\hat{\mathcal{A}}\rep^{-1}(\ket{\widetilde{\Psi}_1})\),
\end{equation}
for any $\ket{\widetilde{\Psi}_1}\in\widetilde{\hilbert}$  (or $\widetilde{\mathcal{V}}$).
The resulting operators are linear on $\widetilde{\hilbert}$ (or $\widetilde{\mathcal{V}}$), and if $\hat{\mathcal{A}}$ is Hermitian or unitary, so is $\rep(\hat{\mathcal{A}})$.

\begin{corollary}
\label{thm:operators_as_fields}
Linear operators on $\mathcal{V}$ admit, via the representation $\rep$, a representation as fields on the $3$D space $\M$. This applies in particular to operators on $\hilbert$.
\end{corollary}
\begin{proof}
Since $\mathcal{V}$ is a vector space, linear operators acting on them are elements of the tensor product $\mathcal{V}\otimes\mathcal{V}^\ast$. But then, we can represent them as local fields on $\M$, just like we did with the tensor product states. Therefore, linear operators have a local field representation.
\end{proof}

Recall that the $\n$-particles states of the same type $\D$ are not simply vectors from the tensor product space $\bigotimes^\n\V{\M,\D}$, but either from its symmetrized or its antisymmetrized (or alternating) tensor product, as in equation \eqref{eq:n_identical}:
\begin{equation}
\tag{\ref{eq:n_identical}'}
\F{\D}^\n{}_\pm := \Sym\pm\(\underbrace{\V{\M,\D}\otimes\ldots\otimes\V{\M,\D}}_{\tn{$\n$ times}}\),
\end{equation}

The symmetry or antisymmetry conditions make sense also in the case of the operation $\btimes$, so fermions and bosons will be represented as fields with the appropriate symmetries,
\begin{equation}
\label{eq:n_identical_field}
\widetilde{\F{\D}}^\n{}_\pm := \Sym\pm\(\underbrace{\V{\M,\D}\btimes\ldots\btimes\V{\M,\D}}_{\tn{$\n$ times}}\),
\end{equation}
where the operator $\Sym+$ symmetrizes the operation $\btimes$, and $\Sym-$ anti-symmetrizes it.
So our Hilbert space $\hilbert$ has to include such fermionic and bosonic states, and the representation \eqref{eq:btimes_space_rep} will take care that the corresponding fields have the right symmetries \eqref{eq:n_identical_field}.

This concludes the proof of Theorem \ref{thm:space_representation}.

%------------------------------------------------------------%
\section{Dynamics and locality}
\label{s:locality}

We have seen that many-particle quantum states can be represented as fields on the $3$D space. These fields are similar to classical fields, but much more complex, in order to represent the degrees of freedom of quantum states. Now we will see that, as long as no measurement occurs and the dynamics is governed only by unitary evolution, their evolution is local in the $3$D space.

The Hamiltonian for $n$ particles in {\NRQM} has the form
\begin{equation}
\label{eq:schrod_hamiltonian}
\hat{H} = -\sum_j\frac{\hbar^2}{2m_j}\nabla_{\x_j}^2 + \sum_{j\neq k}V\(\x_j,\x_k\).
\end{equation}
In general, the potential depends on the $3$D distance between $\x_j$ and $\x_k$, so $V\(\x_j,\x_k\)=V\(|\x_k-\x_j|^2\)$, where $|\x_k-\x_j|$ is the $3$D norm.
Generalizations that include spin or other degrees of freedom can be put in similar form.

In the absence of interactions, there is no indication in the Hamiltonian about the dimension of the space on which the wave function is defined, but the potential indicates three space dimensions, because $V$ depends on $3$D distances  $|\x_k-\x_j|$ rather than $3\n$-dimensional distances in the configuration space. This suggests that the dynamics is in some sense $3$-dimensional, even though the wave function is defined on the configuration space \cite{DavidAlbert1996ElementaryQuantumMetaphysics}.

But now we have, in addition, a representation of the wave function as a multi-layered field on the $3$D space.
The Hamiltonian $\hat{H}$ also has a representation $\rep(\hat{H})$ acting on $\widetilde{\hilbert}$, as defined in equation \eqref{eq:operators_fields}.
Let us first see how it acts on products $\psi_1(\x_1)\ldots\psi_\n(\x_\n)$.
Even if separable states can evolve into nonseparable states, we can still consider the instantaneous value of the state at a time $t$, and focus on product terms of the form $\psi_1(\x_1,t)\ldots\psi_\n(\x_\n,t)$ in the total nonseparable state. In this case, each term of the kinetic part of $\hat{H}$, $-\frac{\hbar^2}{2m_j}\nabla_{j}^2$, acts by differentiating only $\psi_j$. In terms of multi-layers of the form $\(\psi_1\btimes\ldots\btimes\psi_\n\)(\x,t)$, its representation $\rep\(-\frac{\hbar^2}{2m_j}\nabla_{j}^2\)=-\frac{\hbar^2}{2m_j}\widetilde{\nabla}_{j}^2$ acts by 
%\begin{widetext}
\begin{eqnarray}
\label{eq:hamiltonian_kinetic_field_j}
\(-\frac{\hbar^2}{2m_j}\widetilde{\nabla}_{j}^2\(\psi_1\btimes\ldots\btimes\psi_\n\)\)(\x,t) \qquad\qquad\nonumber \\
= -\frac{\hbar^2}{2m_j}\({\psi_1}\btimes\ldots\btimes\nabla_{j}^2{\psi_j}\btimes{\psi_\n}\)(\x,t).
\end{eqnarray}
%\end{widetext}
Since the operator $\nabla$ on fields on the $3$D space is local, the kinetic terms of the Hamiltonian act locally on the $3$D space.

The potential terms $V\(\x_j,\x_k\)$ involve a dual role of the one-particle wave functions composing $\psi_1(\x_1)\ldots\psi_\n(\x_\n)$, since one of them, say $\psi_j(\x_j)$ has the role of the source of the potential, and $\psi_k(\x_k)$ is the one affected by the potential. But by considering an instantaneous superposition of products of the form $\psi_1(\x_1)\ldots\psi_\n(\x_\n)$, we can extract the total potential affecting $\psi_k(\x_k)$ being sourced by all $\psi_j(\x_j)$ with $j\neq k$,
%\begin{widetext}
\begin{eqnarray}
\label{eq:hamiltonian_potential_jk}
\sum_{j\neq k}V\(\x_j,\x_k\)\psi_1(\x_1)\ldots\psi_\n(\x_\n) \qquad\qquad\qquad\qquad\qquad \\
= \psi_1(\x_1)\ldots\(\(\sum_{j\neq k}V\(\x_j,\x_k\)\)\psi_k(\x_k)\)\ldots\psi_\n(\x_\n),\nonumber 
\end{eqnarray}
%\end{widetext}
where we sum only over $j\neq k$ and keep $k$ fixed.

This allows us to separate the effect of the total potential on each $\psi_k$ as $V^{\tn{tot}}(\x_k)\psi_k(\x_k)$, where $V^{\tn{tot}}(\x_k)=\sum_{j\neq k}V\(\x_j,\x_k\)$. Note that while $V^{\tn{tot}}(\x_k)$ is obtained by summing various potentials of the form $V\(\x_j,\x_k\)$, $j\neq k$, this only means that it depends on the positions of the sources, but what matters to the particle represented by $\psi_k(\x_k)$ is the total value at $\x_k$, $V^{\tn{tot}}(\x_k)$.
If there is a source-free field, this can also be included in $V^{\tn{tot}}(\x_k)$, since it already has the form $V(\x_k)$.
Also, the Hamiltonian \eqref{eq:schrod_hamiltonian} approximates the potentials as acting instantaneously, while in fact we should take into account the limit velocity for the interactions, which is $c$. But retarded potentials have the same form $V^{\tn{tot}}(\x_k)$ as well.
The separation of the Hamiltonian per particle, as well as the limited propagation velocity of the interactions, is more evident in the {\schrod}-Pauli and Dirac equations for many-particles interacting electromagnetically, where the momentum term for each particle is supplemented by a term due to the connection, see {\eg} \cite{Crater1983TwoBodyDiracEquations}.

When moving to multi-layered field representations, we no longer need to index $\x$ as $\x_k$ for each particle, but we still need to index the total potential as $\widetilde{V}^{\tn{tot}}_{k}(\x)$, in order to know that it is the potential affecting $\psi_k$ and sourced by all other $\psi_j$ and possible sourceless components. So the representation of the Hamiltonian becomes
\begin{equation}
\label{eq:schrod_hamiltonian_field}
\rep(\hat{H}) = -\sum_j\frac{\hbar^2}{2m_j}\widetilde{\nabla}_{j}^2 + \sum_{k}\widetilde{V}^{\tn{tot}}_{k},
\end{equation}
and acts on $\Psi(\x,t)=\psi_1(\x,t)\btimes\ldots\btimes\psi_\n(\x,t)$ by
%\begin{widetext}
\begin{eqnarray}
\label{eq:schrod_hamiltonian_field_act}
\rep(\hat{H})\Psi(\x,t) \qquad\qquad\qquad\qquad\qquad\qquad\qquad\qquad\qquad\qquad\qquad\\
= \(\sum_{k}{\psi_1}\btimes\ldots\btimes\(\(-\frac{\hbar^2}{2m_j}\nabla_{k}^2+V^{\tn{tot}}_k\){\psi_k}\)
\btimes\ldots\btimes{\psi_\n}\)(\x,t).\nonumber
\end{eqnarray}
%\end{widetext}

Even if the potentials depend on the fields that sourced them, we can see that the Hamiltonian acts pointwisely, and act differently in each layer. The result extends immediately to superpositions of separable states by the linearity of the operator $\rep(\hat{H})$.

Since in fact the potentials propagate locally with limited velocity, it also follows that the unitary time evolution in the field representation is local. 
This may seem at odds with the well-established result that there are correlations in the outcomes of quantum measurements which appear to be nonlocal \cite{Bell64BellTheorem,Aspect99BellInequality}. But such correlations are obtained only by measurements, which seem to require the occurrence of a projection of the state vector normally associated to measurements \cite{vonNeumann1955MathFoundationsQM}. Such a projection would change instantaneously the wave function everywhere, so it would be nonlocal. Now, that we know that even the most highly entangled quantum states can be represented as fields on the $3$D space, it becomes clearer that nonlocality is not due to the fact that the wave function is defined on the configuration space. Nonlocal correlations occur during the measurements. More about this in section \sref{s:EPR}.

%------------------------------------------------------------%
\section{Multi-layered field representation of quantum states}
\label{s:multilayer}

In section \sref{s:st_rep}, we have seen that we can form tuples of wave functions or fields defined on the $3$D space, by using constructions encountered in the theory of fiber bundles, leading to representations of separable states as fields defined on the $3$D space. This is done by using the operation $\btimes$, which is defined using an equivalence class of direct sums of vector fields. Superpositions of separable states are represented as direct sums of fields representing separable states. One not only gets a $3$D space representation of the wave functions defined on the configuration space, but also of the linear operators acting on them, as operators on the fields \eqref{eq:operators_fields}. This construction provides a background for an intuitive interpretation, based on multiple layers. A \emph{layer} consists of the representation of a separable state, and superpositions of such separable states are represented as linear combinations of such layers, which can be called \emph{multi-layers}. Fig. \ref{multi-layered.png} depicts this idea.

Equation \eqref{eq:multi-layered_field} shows how the multi-layered fields like the one in Fig. \ref{multi-layered.png} can be represented as long chains of direct sum of vector fields on $\M$, some of them coupled into layers by transformations $\mathcal{G}_{\nD{\D_1},\ldots,\nD{\D_\n}}$ as in equation \eqref{eq:btimes_transf_n_orbit}.
\begin{equation}
\label{eq:multi-layered_field}
\Psi=(\underbrace{\psi_1,\psi_2,\psi_3}_{\mathcal{G}_{\nD{\D_1},\nD{\D_2},\nD{\D_3}}},\underbrace{\psi_1',\psi_2',\psi_3'}_{\mathcal{G}_{\nD{\D_1},\nD{\D_2},\nD{\D_3}}},\underbrace{\psi_1'',\psi_2'',\psi_3''}_{\mathcal{G}_{\nD{\D_1},\nD{\D_2},\nD{\D_3}}}).
\end{equation}

The representation $\rep$ of operators on the vector space $\hilbert$ as operators on the vector space $\widetilde{\hilbert}$ defined in equation \eqref{eq:operators_fields} applies, in particular, to \emph{creation} and \emph{annihilation} operators, which are used to construct quantum states out of the vacuum state. In this sense, creating and annihilating particles whose states are from the basis \eqref{eq:basis_nonseparable_field} can be seen intuitively as adding and removing particles from the layers, or rather as moving the state from one combination of layers to another one.

The idea behind this model may be, implicitly and informally, behind the intuition of some working physicist, who seem to consider the wave function as defined on the configuration space, but at the same time on the $3$D space. If not, it can be a basis for such an intuition. Mathematical manipulation works perhaps easier in the tensor product formalism, but there are some intuitive hints of the wave function being defined on the $3$D space. First, the dynamics, as explained in section \sref{s:locality}. Then, the measuring apparatus is usually considered implicitly quasi-classical, having all parts well localized in the $3$D space, which suggests that the wave functions of the measured particles are there too. Then, the representation of quantum states by applying combinations of creation and annihilation operators on the representation of the vacuum state can be easily understood as operating on the layers or multiple layers. Nevertheless, a pedagogical emphasis of the configuration space representation as done in \cite{Albert2019How2TeachQM} will retain its importance for the understanding of quantum mechanics.

In general, the layers are not usually conserved by unitary time evolution, because separable states don't remain separable. This happens in particular when interactions are present. However, local interactions are understood in the multi-layered field representation to lead to local dynamics of the fields representing the wave functions.

%------------------------------------------------------------%
\section{Nonlocal correlations. The EPR experiment}
\label{s:EPR}

We have seen that the wave function admits a $3$D space representation, even when entanglement is present, and its dynamics is local as long as only unitary evolution takes place.
This may seem to contradict the existence of nonlocal correlations in quantum mechanics.
In fact it doesn't, because nonlocal correlations appear when quantum measurements are made.
Nonlocal correlations are not due to entanglement alone, but to whatever happens that we call projection of the state vector \cite{vonNeumann1955MathFoundationsQM}, when applied to an entangled state.

Let's see how this works in the EPR experiment \cite{EPR35,Bohm51}. Consider the following state of two spin $1/2$ particles,
\begin{equation}
\label{eq:spin_singlet_config_space}
\psi_A(\x_A,+)\psi_B(\x_B,-) - \psi_A(\x_A,-)\psi_B(\x_B,+),
\end{equation}
where $\psi_j(\x_j,\pm)$, $j\in\{A,B\}$, denote the components of the wave function corresponding to the spin along the $\pm z$ axis. 
Here I use as indices $A$ and $B$, to honor Alice and Bob for their tireless efforts to perform our thought experiments.

The two particles are assumed to go in different places in space, where the spin of one of them is measured by Alice, and the spin of the other by Bob. Then, if both of them measure the spin along the $z$ axis, Alice gets $+\frac12$ and Bob gets $-\frac12$, or vice versa. 
The result is obtained, according to \cite{vonNeumann1955MathFoundationsQM}, by projecting to one of the eigenspaces of the combined spin operator corresponding to the joint measurement, $\hat{\sigma}_z^A\otimes\hat{\sigma}_z^B$.

Let's rewrite \eqref{eq:spin_singlet_config_space} in terms of the multi-layered field representation on the $3$D space $\M$:
\begin{equation}
\label{eq:spin_singlet_physical_space}
\psi_A^+\btimes\psi_B^- - \psi_A^-\btimes\psi_B^+,
\end{equation}
where $\psi_j^\pm(\x)=\psi_j(\x,\pm)$, $j\in\{A,B\}$. 

The combined spin operator $\hat{\sigma}_z^A\otimes\hat{\sigma}_z^B$ translates, via the isomorphism $\rep$, into an operator $\rep\(\hat{\sigma}_z^A\otimes\hat{\sigma}_z^B\)$ on multi-layered fields, \emph{cf.} equation \eqref{eq:operators_fields}. Then,
\begin{equation}
\label{eq:EPR_spin_operator_fields}
\rep\(\hat{\sigma}_z^A\otimes\hat{\sigma}_z^B\) = \widetilde{\sigma_z^A}\btimes\widetilde{\sigma_z^B}.
\end{equation}
Here, $\widetilde{\sigma_z^A}$ acts on the first sublayer of each layer, and $\widetilde{\sigma_z^B}$ acts on the second sublayer of each layer.
Their eigenstates select the fields $\psi_A^+\btimes\psi_B^-$ and $\psi_A^-\btimes\psi_B^+$.
The projection postulate requires that the two particles are found either in one state, or the other.

Now, the measurement is about determining whether the observed particles are either in the layer corresponding to $\psi_A^+\btimes\psi_B^-$, or in the one corresponding to $\psi_A^-\btimes\psi_B^+$. If we see the EPR experiment, intuitively, as taking place in the $3$D space, it is about the locations of the observed particles in one layer or another, and the multi-layered field representation gives a support for this intuition.

If the spin measurements are done by orienting the Stern-Gerlach devices along different directions, the resulting layer will be ``oblique'' with respect to the layers in equation \eqref{eq:spin_singlet_physical_space}, which is guaranteed by the isomorphism $\rep$.

In general, the projection operators corresponding to the possible outcomes, being linear operators on the Hilbert space $\hilbert$, have, via the isomorphism $\rep$, corresponding projectors on the space of multi-layered fields $\widetilde{\hilbert}$. The projection postulate translates in general in selecting a layer for the state of the observed system. Quantum correlations are obtained exactly as in the standard representation of quantum states, due to the isomorphism $\rep$ between quantum states and multi-layered fields on the $3$D space, and between the operators on quantum states and operators on fields (section \sref{s:st_rep}).

%------------------------------------------------------------%
\section{Primitive ontology of the wave function}
\label{s:implications}

The construction presented in this paper is just a representation of the quantum states, in terms of fields on the $3$D space, rather than in terms of wave functions on the $3\n$-dimensional configuration space. As such, while it is not committed to any interpretation of quantum mechanics or ontology, it is able to provide a primitive ontology for the wave function. 

The ontic position about the wave function is endorsed by some results and theorems \cite{Spekkens2005Contextuality,HarriganSpekkens2010EinsteinIncompleteness,ColbeckRenner2011NoExtensionOfQM,ColbeckRenner2012Reality,PBR2012RealityOfPsi,Hardy2013AreQuantumStatesReal,Ringbauer2015MeasurementsRealityWavefunction,Myrvold2018PsiOntology}. 
However, the fact that we represent wave functions on the configuration space was often regarded as a sign that the wave function is not a real physical thing, as shown in section \sref{s:intro}. 
 
The wave function is taken to be ontic in the \emph{many worlds interpretation} (MWI) \cite{Eve57,Eve73,dWEG73,SEP-Vaidman2002MWI,Saunders2010ManyWorlds,Pas2017MWI,Marchildon2017SpacetimeInMWI,carroll2019MadDogEverettianism}, where the definiteness of outcomes is explained by the fact that two wave functions in superposition ignore one another. 

Similarly, in some spontaneous or objective \emph{collapse theories} \cite{GRW86,Ghirardi1990RelativisticDynamicalReductionModels,Pearle1989StochasticSpontaneousLocalization,penrose1996gravityQuantumStateReduction}, the wave function is taken as ontic, as well as its collapse, while in the \emph{flash ontology} this is avoided.
 
Even in Bohm's version of the {\pwt} \cite{Bohm52,Bohm95,Bohm2004CausalityChanceModernPhysics} the wave function is considered real~\footnote{Apparently, Bell endorsed this view on the {\pwt}: \emph{``No one can understand this theory until he is willing to think of [the wave function] as a real objective field rather than just a `probability amplitude'. Even though it propagates not in 3-space but in 3N-space''} \cite{Bell2004SpeakableUnspeakable} p. 128.}, although there are versions which try to avoid this, like the \emph{nomological} interpretation of the wave function (where the wave function is interpreted as a physical law prescribing a nonlocally coordinated motion of the point-particles in the $3$D space) \cite{Durr1995BohmianWavefunction,GoldsteinTeufel2001QuantumSpacetimeWithoutObservers,GoldsteinZanghi2013RealityAndRoleOfWavefunction}~\footnote{Also see the \emph{Humeanist interpretation}, for a weaker version of the nomological position \cite{Loewer1996HumeanSupervenience,Loewer2001DeterminismAndChance,Hall2015HumeanReductionism,Miller2014QEntanglementBMHumeanSupervenience,Esfeld2014QHumeanism,Bhogal2017HumeanismEntanglement,Callender2015OneWorldOneBeable,EsfeldDeckert2017MinimalistOntologyOfNaturalWorld}.}.

MWI has been recently criticized for not having a $3$D space or spacetime ontology \cite{Maudlin2010CanTheWorldBeOnlyWavefunction,Norsen2017FoundationsQM}. A similar situation is present in the collapse theories, where the wave function collapses. Also in Bohm's version of the {\pwt}, even though there the wave function plays a different role.

An usually adopted primitive ontology of the wave function in such interpretations, originating with {\schrod}, is the charge or mass \emph{density ontology}. It consists in considering the total charge or mass density of the universal wave function, which is a function on the $3$D space, as the ontology. This solution is used for example in collapse theories (where it is called GRWm, to be distinguished from the flash ontology called GRWf), and in MWI \cite{SEP-Vaidman2002MWI}, and it is satisfactory to some extent. Unfortunately, this kind of ontology misses most of the information encoded in the universal wave function.
The position that the wave function requires more than the mass density ontology is endorsed for example in \cite{DavidAlbert1996ElementaryQuantumMetaphysics,Maudlin2007CompletenessSupervenienceOntology,Maudlin2013TheNatureOfTheQuantumState,Dewar2016LaBohume,Maudlin2019PhilosophyofPhysicsQuantumTheory}.
In addition, in the case of GRWm, it seems to be at odds with relativistic simultaneity, for which a more relativistic invariant modification was proposed \cite{Bedingham2014MatterDensityAndRelativisticWavefunctionCollapse,Tumulka2006RelativisticGRW}.

An improved $3$D space ontology for the wave function, compared to the mass density ontology, is the \emph{space state realism}, suggested in \cite{WallaceTimpson2010QMOnSpacetime}. There, more information about the total wave function is gained by using the reduced density matrices. Also see \cite{Swanson2018RelativisticSpacetimeStateRealist} for an extension to the relativistic case and quantum field theory. While it is an improvement in the amount of information from the wave function that it represents, it still captures a very small part of it. It was criticized in \cite{Maudlin2019PhilosophyofPhysicsQuantumTheory}, mainly for not providing a monistic ontology, for mixing together different branches (if applied to MWI), and for lack of local separability, due to the use of reduced density matrices. This seems to make unlikely the existence of local beables able to encode the full information about the state. 

For MWI, {\pwt}, and collapse theories, the multi-layered field representation proposed here allows to take the full universal wave function as primitive ontology. In contrast to {\schrod}'s density ontology and even to space state realism, a multi-layered field ontology would retain all the information in the wave function.
Nonseparability goes away, since we can recover the multi-layered field over a region $A\cup B$ by knowing it on $A$ and $B$. This is possible because the representation keeps track of the layers and sublayers, and the extension from $A$ and $B$ is done by connecting each sublayer and layer over $A$ to those over $B$ (see \sref{s:local_separability} and \sref{s:locality}).
In the case of collapse theories, a primitive ontology based on multi-layered fields would still be a tension with relativistic simultaneity. However, one should not exclude the possibility of an adaptation of the proposal in \cite{Bedingham2014MatterDensityAndRelativisticWavefunctionCollapse,Tumulka2006RelativisticGRW} to this ontology.

A wave function spacetime ontology would be in particular useful to those approaches trying to save relativistic invariance and locality at the expense of \emph{statistical independence} of the states to be observed from the observation to be made, while still being able to get the nonlocal correlations. 
Saving locality is not actually forbidden by Bell's theorem \cite{Bell64BellTheorem}, because the theorem relies on two assumption, locality and statistical independence, to derive Bell's inequality \cite{Maudlin1996SpacetimeInTheQuantumWorld,Gill2014StatisticsCausalityBellTheorem}. Experimental observations of violations of Bell's inequality \cite{Aspect99BellInequality} imply only that at least one of these hypotheses should be rejected, but there is the option to reject statistical independence and keep locality, as these models show. Such models thus involve a dependence of past events on future events, sometimes called \emph{retrocausality} \cite{deBeauregard1977TimeSymmetryEinsteinParadox,SEP-2019qm-retrocausality,Rietdijk1978retroactiveInfluence,Wharton2007TimeSymmetricQM,price2008toyRetrocausality,sutherland2008causallySymmetricBohm,Argaman2008Retrocausality,price2015disentangling,sutherland2017RetrocausalityHelps,EmilyAdlam2018SpookyActionTemporalDistance,CohenCortesElitzurSmolin2019RealismAndCausalityI,WhartonArgaman2019BellThereomAndSpacetimeBasedQM}. This path is also taken in the \emph{transactional interpretation}, which also provides a mechanism of negotiation taking place not in the physical time, but in a \emph{pseudotime} \cite{cramer1986transactional,cramer1988overview,Kastner2012TransactionalBook}.
Such theories may be able to avoid nonlocality as ``action at a distance'', by relying to get the nonlocal correlations either on very special, seemingly ``conspirational'' initial conditions, or on an apparent zig-zag of local causal influence back and forward in time. Another interesting proposal is to take spacetime as a context for quantum measurements, giving by this a local account of the original EPR experiment \cite{Khrennikov2002LocalRealismContextualismLoopholesBell,Khrennikov2009ContextualApproachToQuantumFormalism} by using contextuality. Contextuality is required by the Kochen-Specker theorem \cite{Bel66,KochenSpecker1967HiddenVariables,AbbotCaludeSvozil2015-AnalyticKS,Loveridge2015ManyMathematicalMerminKochenSpecker}.
Such approaches have the advantage of being more consistent with relativistic invariance~\footnote{Lorentz invariance seems more difficult to be satisfied by collapse theories and {\pwt}, but such proposals exist \cite{Bedingham2014MatterDensityAndRelativisticWavefunctionCollapse,Tumulka2006RelativisticGRW,DurrEtAl2014BohmianRelativistic,sutherland2008causallySymmetricBohm,sutherland2017RetrocausalityHelps}.}.

There are also proposals based on unitary evolution of the wave function (for single worlds), without any real collapse, like Schulman's \emph{special states} approach \cite{schulman1984definiteMeasurements,schulman1991definiteQuantumMeasurements,schulman1997timeArrowsAndQuantumMeasurement,schulman2016specialStatesDemandForceObserver} and references therein, 't Hooft's \emph{cellular automaton interpretation} \cite{tHooft2011CollapseSchrodingerBornRule,Elze2006DeterministicModelsForQuantumFields,Elze2014Action4CellularAutomata,tHooft2016CellularAutomatonInterpretationQM}, and a proposal based on \emph{global consistency} \cite{Sto08b,Sto12QMb,Sto12QMc,Sto16aWavefunctionCollapse,Sto16UniverseNoCollapse,Sto2019PostDetermined}~\footnote{This is proposed to work in terms of gluing local solutions into global solutions, as in \emph{sheaf theory} \cite{MacLaneMoerdijk92,bredon1997sheaf,ROWells2008ComplexManifolds}. Constraints, mainly topological in nature, prevent most local solutions or initial conditions to be extended globally, which leads to a prevention of statistical independence. Sheaf theory and its sibling \emph{topos theory} were already applied in the foundations of quantum mechanics, in particular to contextuality, see  \cite{DoringIsham2008ToposFoundationsPhysicsI2IV,AbramskyBrandenburger2011sheafNonlocalityContextuality,AbramskyConstantin2014ClassificationMultipartiteStatesNonlocality,Constantin2015SheafQMThesis,CeciliaFlori2012ToposQuantumTheory} and references therein.}. 
Such theories make a clear prediction, identified in \cite{schulman2016specialStatesDemandForceObserver} -- a particle prepared in a spin eigenstate along some axis, when having its spin measured again, will exhibit a force required to reorient the spin iff the new axis differs from the previous one. The reason is that, in such models, we are not allowed to turn an ontic state into a superposition of ontic states.
Superpositions of ontic states are, in these approaches, epistemic.
Another prediction of this class of theories is that the conservation laws are \emph{not} violated, while collapse or branching should violate them \cite{Sto16UniverseNoCollapse}, which so far was never found to be wrong.
These predictions are probably very difficult to test experimentally, but if confirmed, approaches based on collapse or branching can be ruled out. 
Presumably more such predictions can be made, for example concerning the spacetime curvature due to the quantum states, but this is probably inaccessible to our experimental capabilities.
In pursuing locality, all such approaches would benefit from a $3$D space ontology provided by the multi-layered fields, especially since their unitary dynamics is local, as shown in section \sref{s:locality}.

Aside from providing a $3$D space ontology and locality for the unitary dynamics of the wave function, the multi-layered field representation is not enough to solve the measurement problem and the problem of the emergence of the quasi-classical world. The representation doesn't provide the sort of beables able to determine the outcomes of measurements, or the way the universal wave function branches. It also doesn't eliminate {\schrod} cats. Solving these problems should be done in conjunction with additional hypotheses of the types proposed by the various interpretations of quantum mechanics.

Regardless of which interpretation of quantum mechanics is the right one or at least the preferred one by the reader, the existence of a $3$D space representation of the wave function can be helpful to those approaches taking the wave function as ontic, but also to those taking it as epistemic or nomological.

%------------------------------------------------------------%
\section{Possible objections}
\label{s:objections}

When developing the multi-layered field representation, I tried to submit it to various personal objections, and to anticipate potential objections from the readers. Here are some of them, that I considered more relevant or likely to be raised.

\begin{objection}
The major claim of the paper sounds interesting, but it is a too long reading, and I will not invest time in something known to be impossible.
\end{objection}
\begin{reply}
I suggest, before you decide whether to read it carefully, to check Figure \ref{multi-layered.png}.
\end{reply}

\begin{objection}
This representation is equivalent to the configuration space representation. How does this help? Doesn't it mean that the configuration space remains?
\end{objection}
\begin{reply}
The configuration space remains. The objective was not to remove it, and we can't remove it, because it is inherent to quantum mechanics. The situation is similar to classical mechanics, where one can represent the particle configurations both in the $3$D space, and in the $3\n$-dimensional configuration space, and also in the $6\n$-dimensional phase space.
\end{reply}

\begin{objection}
For a theory to be scientific, it has to be falsifiable. Are there any empirical predictions of your model?
\end{objection}
\begin{reply}
The model constructed here is just a representation, equivalent to the {\schrod} representation of wave functions on the configuration space, but in terms of fields on the $3$D space. For this reason, it should not be expected to give different predictions from quantum theory, because it is not a different theory. 
\end{reply}

\begin{objection}
It seems unlikely that the universe is as complicated as your construction.
\end{objection}
\begin{reply}
The role of this representation is merely to provide a proof of concept that the universal wave function can be an object in the physical space. Whether this is effectively realized in nature, and whether it is realized in this form or another, it doesn't say.
\end{reply}

\begin{objection}
The fiber bundle structure in your representation is too large!
\end{objection}
\begin{reply}
It is very large indeed, because it needs to be able to represent wave functions defined on a space with infinitely many dimensions.
\end{reply}

\begin{objection}
There are no foundational open problems in quantum mechanics, or they are solved by interpretation X. Your construction is pointless.
\end{objection}
\begin{reply}
If your favorite interpretation X requires the wave function to be ontic, then the multi-layered field representation may help (section \sref{s:implications}). If in your favorite interpretation X the wave function is epistemic or nomological, it is not hurt by this representation.
\end{reply}

\begin{objection}
You claimed that the wave function can be understood as  defined on the $3$D space. But we know from the EPR experiment that there is entanglement. This disproves your theory.
\end{objection}
\begin{reply}
The multi-layered field representation is capable to represent all possible quantum states in {\NRQM}, including entangled systems (section \sref{s:nonseparable}).
For a discussion of the EPR see section \sref{s:EPR}.
\end{reply}

\begin{objection}
If your representation is able to represent all possible quantum states in {\NRQM}, including entangled systems, then doesn't this mean that it predicts {\schrod} cats?
\end{objection}
\begin{reply}
It makes the same predictions as quantum mechanics, because it is quantum mechanics. So yes, unfortunately it also predicts {\schrod} cats, but it is not the objective of this representation to solve this problem. But this representation can be part of the ontology in some interpretations trying to solve it (see section \sref{s:implications}).
\end{reply}

\begin{objection}
A $3$D space representation of the wave function can't exist, because it would violate nonlocality.
\end{objection}
\begin{reply}
Not only it does exist, but its unitary dynamics is local, at least as long as no measurements are involved (see section \sref{s:locality}). This locality is in no conflict with Bell's theorem, because it is true only as long as only unitary evolution happens. If measurements are involved, nonlocal correlations do appear though (section \sref{s:EPR}).
\end{reply}

\begin{objection}
Your representation can't give an ontology to the wave function, because an ontology should also include beables that solve the measurement problem.
\end{objection}
\begin{reply}
The purpose of this representation is merely to prove the possibility that the wave function can be understood as existing in the $3$D space. It doesn't solve, and I don't claim it solves, the measurement problem. This should be done in conjunction with other theories or interpretations of quantum mechanics (see section \sref{s:implications}).
\end{reply}

\subsection*{Acknowledgement}
The author thanks Eliahu Cohen, Hans-Thomas Elze, Art Hobson, Louis Marchildon, Travis Norsen, Larry Schulman, Lev Vaidman, Sofia Wechsler, and Ken Wharton, for their valuable suggestions offered to a previous version of the manuscript. Nevertheless, the author bares full responsibility for the article.

\appendix

%------------------------------------------------------------%
\section{Vector bundles}
\label{s:vbundle_math}

Historically, vector fields were regarded are functions defined on a space $M$ and valued in a fixed vector space $V$. But since this could not work well in all situations, for example if the base space $M$ is topologically nontrivial, or when the field does not consist of vectors from $M$ (in the case when $M$ itself can be seen as a vector space), the idea had to be made invariant and generalized. This led to the necessity to associate a distinct copy $V_x$ of the vector space $V$ at each point $x\in M$, and also to specify how this construction can be global on $M$ in a continuous way. This led to the notion of vector bundle, which consist of the continuous union of all copies $V_x$ of the vector space $V$ at all points $x\in M$.

I very briefly give the more general and rigorous definition of a vector bundle. More about this rich topic can be found for example in \cite{NashSen1983TopologyGeometryPhysicists,Nakahara2003GeometryTopologyPhysics,Bleecker2005GaugeTheoryVariationalPP,ROWells2008ComplexManifolds}.

\begin{definition}
\label{def:vector_bundle_math}
A complex \emph{vector bundle} $V \to E \stackrel{\pi}{\to} M$, of \emph{rank} $k$, where $k\in\{1,2,\ldots,\infty\}$, is defined by
\begin{enumerate}
	\item 
	A $k$-dimensional vector space $V$ called the \emph{typical fiber}.
	\item 
	Two topological spaces: a \emph{base space} $M$, and a \emph{total space} $E$.
	\item 
	A continuous surjection $\pi:E\to M$, called \emph{bundle projection}, so that for every $x\in M$, the \emph{fiber} over $x$, $\pi^{-1}(x)$ is a $k$-dimensional complex vector space $V_x$ isomorphic to $V$.
	\item 
	The following \emph{compatibility condition}: for every $x\in M$, there is an open neighborhood $U$ of $x$, and a homeomorphism $\varphi_U:U\times\C^k\to\pi^{-1}(U)$ so that for all points $y\in U$ and vectors $v\in\C^k$,
\begin{enumerate}
	\item 
	$(\pi\circ\varphi_U)(y,v)=y$, and
	\item 
	the map $v\mapsto\varphi_U(y,v)$ is a vector space isomorphism between $\C^k$ and $V_y=\pi^{-1}(y)$.
\end{enumerate}
\end{enumerate}
\end{definition}

%------------------------------------------------------------%
\section{Multi-layered field representation of quantum fields}
\label{s:quantum_fields}

We have seen that the wave function in {\NRQM} can be represented in terms of fields on the $2$D-space. This representation was made in general enough settings to include spin, internal degrees of freedom, and all entangled states allowed in {\NRQM}. However, it is important to see if it can be scaled up to quantum field theories. Apparently, classical fields have much more degrees of freedom than classical systems of point particles. In addition, quantum field theory is relativistic, as opposed to {\NRQM}.

When we limit to a fixed $\n$, especially a finite value of it, the number of degrees of freedom is $3\n$, and seems ``small'' compared to the continuous range of degrees of freedom of a classical field. However, many-particle systems are unexpectedly rich, because the Fock space has an uncountable number of dimensions even if the one-particle Hilbert space would be separable. This is because of the following reason. The exterior algebra of an $n$-dimensional vector space has $2^n$ dimensions. The symmetric tensor algebra has even more, being infinite-dimensional. So, the Fock spaces for both fermions and bosons are rich enough to represent the continuous set of degrees of freedom needed in field quantization.

In order to see how we can apply the multi-layered field representation to quantum fields, we take the example of a scalar field, solution of the Klein-Gordon equation.
Following for example \cite{Srednicki2007QFT}, we start with a classical scalar field $\varphi(\x,t)$, which is defined on $\M\times\R$, and is valued in $\C$. Since the theory is relativistic, $\M$ is the $3$D-space obtained by fixing a timelike vector (representing the time direction) in the Minkowski spacetime. Then, we expand the classical field $\varphi$ in plane waves:
\begin{equation}
\label{eq:KG_expansion}
\varphi(\x,t)=\int_{\R^3}\frac{\de \mathbf{k}}{(2\pi)^3\sqrt{\omega_\mathbf{k}}}\(a(\mathbf{k})e^{-i\omega_\mathbf{k}t+i\mathbf{k}\cdot\x} + a^\ast(\mathbf{k})e^{i\omega_\mathbf{k}t-i\mathbf{k}\cdot\x}\),
\end{equation}
where $\omega_\mathbf{k}=\sqrt{\abs{\mathbf{k}}^2 + m^2}$, $\mathbf{k}\in\R^3$ is the wave vector, and the dot product in $\mathbf{k}\cdot\x$ is in $\R^3$.

To quantize the field $\varphi(\x,t)$, we promote the Fourier coefficients $a$ and $a^\ast$ to operators $\hat{a}$ and $\hat{a}^\dagger$ satisfying the commutation relations
\begin{equation}
\label{eq:KG_CCR}
\begin{cases}
\left[\hat{a}(\mathbf{k}),\hat{a}(\mathbf{k'})\right] = \left[\hat{a}^\dagger(\mathbf{k}),\hat{a}^\dagger(\mathbf{k'})\right] = 0,\\
\left[\hat{a}(\mathbf{k}),\hat{a}^\dagger(\mathbf{k'})\right] = (2\pi)^3\delta(\mathbf{k}-\mathbf{k'}).
\end{cases}
\end{equation}

Then, $\hat{a}^\dagger(\mathbf{k})$ and $\hat{a}(\mathbf{k})$ create and annihilate scalar particles.

By starting from the vacuum state $\ket{0}$ and applying $\hat{a}$ and $\hat{a}^\dagger$, we construct the bosonic Fock space as being generated by the basis
\begin{equation}
\label{eq:KG_Fock}
\ket{\mathbf{k}_1,\ldots,\mathbf{k}_\n} := \hat{a}^\dagger(\mathbf{k}_1)\ldots\hat{a}^\dagger(\mathbf{k}_\n)\ket{0}.
\end{equation}
But we already know how to construct the multi-layered field representation for this Fock space from section \sref{s:st_rep}.
So, there are no new difficulties when moving to quantum fields.

%------------------------------------------------------------%
%\bibliographystyle{apsrev4-2}
%\bibliography{../../bib/references}
%apsrev4-2.bst 2019-01-14 (MD) hand-edited version of apsrev4-1.bst
%Control: key (0)
%Control: author (72) initials jnrlst
%Control: editor formatted (1) identically to author
%Control: production of article title (-1) disabled
%Control: page (0) single
%Control: year (1) truncated
%Control: production of eprint (0) enabled
%

\end{document}